\documentclass[lettersize,journal]{IEEEtran}
\usepackage{indentfirst}
\usepackage[dvips]{graphicx}
\usepackage{amsfonts}
\usepackage{multirow}
\usepackage{amsmath,amsthm,amssymb}
\usepackage{color,changepage}
\usepackage{algorithm}
\usepackage{algorithmic}
\usepackage{bbm}
\usepackage{verbatim}
\usepackage{makecell}
\usepackage{subfigure}
\usepackage{mathrsfs}
\usepackage{arydshln}
\usepackage{cases}
\usepackage{threeparttable}
\usepackage{extarrows}
\usepackage{array}
\usepackage{bm}
\usepackage{pifont}
\usepackage[T1]{fontenc}
\usepackage{threeparttable}
\usepackage{booktabs}
\usepackage{mdwlist}
\usepackage{xr}
\usepackage{hyperref}
\usepackage{cite}
\allowdisplaybreaks[4]

\definecolor{newcolor}{rgb}{0.5,0,1}

\newcommand{\red}[1]{{\color{red}#1}}

\newcolumntype{I}{!{\vrule width 2pt}}

\newtheorem{Theorem}{Theorem}

\newtheorem{Lemma}{Lemma}

\theoremstyle{remark}
\newtheorem{Remark}{Remark}

\DeclareMathAlphabet{\mathpzc}{OT1}{pzc}{m}{it}

\definecolor{newcolor}{rgb}{0.5,0,1}

\def\BibTeX{{\rm B\kern-.05em{\sc i\kern-.025em b}\kern-.08em
    T\kern-.1667em\lower.7ex\hbox{E}\kern-.125emX}}

\let\emptyset\varnothing

\begin{document}

\title{
Generalized Lagrange Coded Computing: A Flexible Computation-Communication Tradeoff for Resilient, Secure, and Private Computation
}

\author{Jinbao Zhu, Hengxuan Tang, Songze Li, and Yijia Chang
\thanks{
A part of this work was presented at the 2022 IEEE International Symposium on Information Theory \cite{zhu2022generalized}. 
This work was supported in part by the National Natural Science Foundation of China (NSFC) under Grant 62106057, and was supported in part by the Fundamental Research Funds for the Central Universities under Grant 2242024k30059 and Grant 2682024CX030.

Jinbao Zhu and Hengxuan Tang 
are with the Information Coding and Transmission (ICT) Key Laboratory of Sichuan Province, Southwest Jiaotong University, Chengdu 611756, China
(e-mail: jinbaozhu@swjtu.edu.cn; hxuantang@my.swjtu.edu.cn).

Songze Li is with the School of Cyber Science and Engineering, Southeast University, Nanjing 210003, China, and with the Engineering Research Center of Blockchain Application, Supervision and Management (Southeast University), Ministry of Education (e-mail: songzeli@seu.edu.cn).

Yijia Chang is with the Thrust of Internet of Things, The Hong Kong University of Science and Technology (Guangzhou), Guangzhou 510006, China (e-mail: ychang847@connect.hkust-gz.edu.cn).
}}

\maketitle

\begin{abstract}
We consider the problem of evaluating arbitrary multivariate polynomials over a massive dataset containing multiple inputs, on a distributed computing system with a master node and multiple worker nodes. Generalized Lagrange Coded Computing (GLCC) codes are proposed to simultaneously provide resiliency against stragglers who do not return computation results in time, security against adversarial workers who deliberately modify results for their benefit, and information-theoretic privacy of the dataset amidst possible collusion of workers. GLCC codes are constructed by first partitioning the dataset into multiple groups, then encoding the dataset using carefully designed interpolating polynomials, and sharing multiple encoded data points to each worker, such that interference computation results across groups can be eliminated at the master. Particularly, GLCC codes include the state-of-the-art Lagrange Coded Computing (LCC) codes as a special case, and exhibit a more flexible tradeoff between communication and computation overheads in optimizing system efficiency. 
Furthermore, we apply GLCC to distributed training of machine learning models, and demonstrate that GLCC codes achieve a speedup of up to $2.5\text{--}3.9\times$ over LCC codes in training time, across experiments for training image classifiers on different datasets, model architectures, and straggler patterns.
\end{abstract}


\section{Introduction}
\IEEEPARstart{A}{s the} era of Big Data advances, distributed computing has emerged as a natural approach to speed up computationally intensive operations by dividing and outsourcing the computation among many worker nodes that operate in parallel.  However, scaling out the computation across distributed workers leads to several fundamental challenges, including additional communication overhead compared to centralized processing, and slow or delay-prone worker nodes that can prolong computation execution time, known as straggler effect~\cite{Tail1,Tail3}. Furthermore, distributed computing systems are also much more susceptible to adversarial workers that deliberately modify the computations for their benefit, and raise serious privacy concerns when processing sensitive raw data in distributed worker nodes. Therefore, designing computation and communication efficient protocols that are resilient against the straggler effect and secure against adversarial workers, while providing a privacy guarantee is of vital importance for distributed computing applications.

Coded distributed computing is an emerging research direction that develops information-theoretic methods
to alleviate the straggler effect, provide robustness against adversarial workers, and protect data privacy, via carefully adopting the idea of error control codes to inject computation redundancy across distributed workers~\cite{li2020coded}.
Coding for distributed computation was earlier considered in \cite{Lee} for some linear function computations (e.g., computing matrix-vector multiplication \cite{li2016unified,mallick2020rateless} or convolution of two input vectors \cite{dutta2017coded}). Subsequently, polynomial codes \cite{Polynomialcode,EPcode} and MatDot codes \cite{MatDotcode} were introduced for mitigating straggler effect and providing privacy guarantee \cite{Tandonsecurecode,Kakar_secure_code,d2020gasp,ZhuSDMM,zhu2020secure,zhu2022systematic} by leveraging the algebraic structure of polynomial functions, within the context of distributed matrix multiplication. 
In a recent work \cite{jia2021cross}, multilinear map codes were developed to compute the evaluations of a multilinear map function over a batch of data,  minimizing communication and computation overheads while mitigating the effect of stragglers. 
Gradient coding \cite{GradientCodes1,GradientCodes2,wang2019erasurehead,li2018near} is also an interesting technique of straggler mitigation for training complicated models with large data sets in a distributed system.
A fundamental tradeoff between computation and communication overheads was established in \cite{LiMapreduce,li2017scalable} for the generally distributed computing frameworks like MapReduce.

\begin{figure}[htbp]
\centering
    \includegraphics[width=7.8cm]{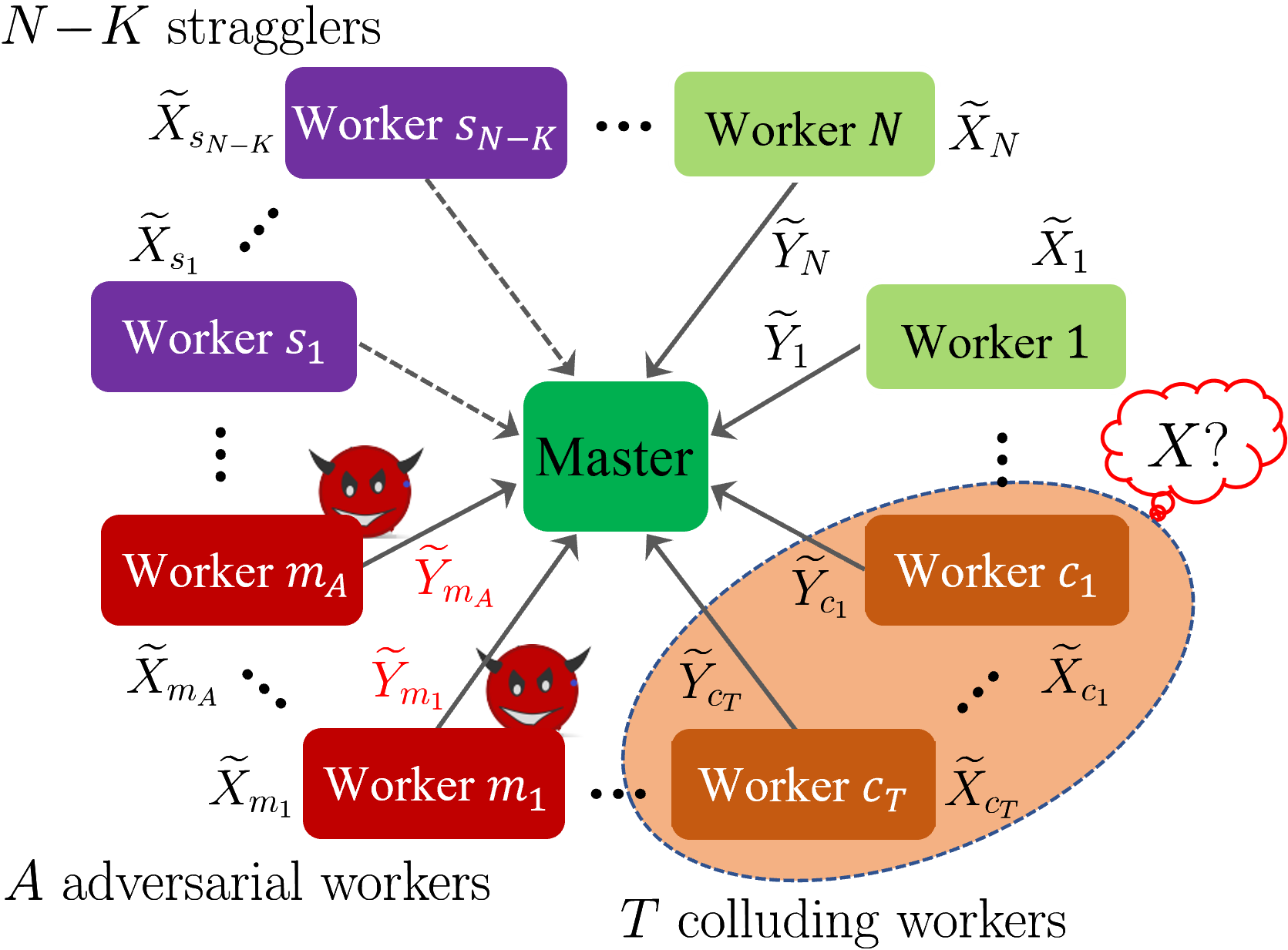}
    \caption{A distributed computing system consisting of a master and $N$ workers for evaluating multivariate polynomials on dataset $X$. Worker $n$ computes a response $\widetilde{Y}_n$ to the master on local data $\widetilde{X}_n$. The master waits for the results from the fastest $K$ workers, $A$ out of whom may be malicious, to recover computation results. Up to $T$ workers may collude to infer about $X$.}
    \label{GLCC}
\end{figure}

This paper considers the distributed computing problem of evaluating arbitrary multivariate polynomial, which was first introduced by Yu \emph{et al.} in \cite{LCC}. 
Specifically, in a distributed computing system including a master node and a set of $N$ worker nodes, the master wishes to compute the evaluation $Y_m=\phi(X_m)$ for each data point $X_m$ in a large dataset $X=(X_1,X_2,\ldots,X_M)$ with the assistance of the $N$ worker nodes, where $\phi$ is an arbitrary multivariate polynomial, as illustrated in Fig. \ref{GLCC}.
It is guaranteed that any up to $T$ colluding workers learn nothing about the dataset $X$ and that the desired polynomial computations must be completed in the presence of $A$ adversarial workers who deliberately modify computation results.
To do that, the master sends an encoded version of $X$ to each worker, who accordingly generates a response according to the encoded data received. Due to the straggler effect, the master only waits for the responses from a subset of the fastest $K$ workers to recover desired evaluations, where the minimum number of successful computing workers that the master needs to wait for is referred to as \emph{recovery threshold}.
Note from references \cite{EPcode,Qian_Yu} that matrix multiplication can be converted into the problem of computing the element-wise products of two batches of sub-matrices based on the concept of bilinear complexity \cite{Strassen,Smirnov}. Moreover, the evaluation of a multivariate polynomial can be also regarded as computing a linear combination of various evaluations of a multilinear map function. Thus, the problem of computing polynomial evaluations over a batch of datasets covers many computations of interest in machine learning, including the aforementioned linear computation \cite{Lee},  matrix multiplication \cite{EPcode,Qian_Yu}, and multilinear map computation \cite{jia2021cross}.

Lagrange Coded Computing (LCC) codes \cite{LCC} were initially proposed for solving the polynomial computation problem, 
but they exhibit low straggler tolerance and require a high recovery threshold.
The goal of this paper is to enhance the straggler tolerance of LCC codes and establish a flexible tradeoff among recovery threshold, communication overhead, and computation overhead to enable optimal designs on system performance. Here, the communication overhead includes the upload cost from the master to the workers and the download cost from the workers to the master, while the computation overhead includes the complexities of generating the encoded data of the dataset at the master, generating the response at each worker, and decoding the desired evaluations at the master.

The main contribution of this paper is a proposed Generalized LCC (GLCC) codes for the polynomial computation problem, which include LCC codes as a special case and establish a more flexible tradeoff among recovery threshold, communication overhead, and computation overhead.
The key components of GLCC codes include 1) partitioning the dataset into \emph{multiple groups}, and encoding the data in each group using Lagrange interpolating polynomial; 2) sharing \emph{multiple evaluations} of each group polynomial with each worker to create computational redundancy; 3) each worker returning multiple computation results to the master, constructed using additional polynomials as interference cancellation coefficients to eliminate interference between groups at the master. Consequently, the computation results at workers can be viewed as evaluations of a composition polynomial, such that the decoding is accomplished by interpolating the polynomial. 

To demonstrate the performance gains over LCC codes, we apply the proposed GLCC codes on distributed training of perceptron neural networks, and empirically evaluate the performance of GLCC codes on simultaneously training multiple single-layer perceptron classifiers. More specifically, in experiments involving various combinations of datasets and straggler scenarios, GLCC codes reduce the training time for the classifiers to reach certain test accuracy by a factor of up to $2.5\text{--}3.9\times$, compared with the LCC codes. Experimental results also demonstrate significant advantages of GLCC codes over LCC codes, in speeding up the model convergence and providing stronger protection on data and model privacy.

\noindent\textbf{Related Work:} 
Since the introduction of LCC codes, there are some works focused on improving the robustness of LCC codes against adversarial workers \cite{soleymani2021list,tang2021verifiable}, applying LCC codes to train machine learning models \cite{so2021codedprivateml,so2020scalable,so2022lightsecagg,shao2022dres}, and implementing LCC codes in analog domains \cite{soleymani2021analog}. More specifically, reference \cite{soleymani2021list} asymptotically improves adversarial tolerance by a factor of two using the folded Reed-Solomon (RS) list decoding ideas in \cite{guruswami2013linear}, but requires the master to perform additional evaluation computations to gain side information about the computation results. Reference \cite{tang2021verifiable} also enhances adversarial tolerance by a factor of two using the concept of verifiable computing to ensure computational integrity. However, it additionally requires the master to compute private verification keys and to verify the correctness of the computation results returned by the workers. The computation overhead of generating verification keys and verifying the computation results depends on the dataset and the multivariate polynomial, and thus is non-negligible.
In recent years, LCC codes have been exploited to train machine learning models in diverse distributed environments while preserving data privacy by quantizing the data into a finite field. These environments include centralized distributed computing systems \cite{so2021codedprivateml}, fully-decentralized distributed systems \cite{so2020scalable}, and federated learning systems \cite{so2022lightsecagg,shao2022dres}.
Unlike these works \cite{soleymani2021list,tang2021verifiable,so2021codedprivateml,so2020scalable,so2022lightsecagg,shao2022dres} that use Shamir's secret sharing \cite{Shamir} to provide information-theoretic privacy guarantees and therefore require computations to be performed in the finite field, reference \cite{soleymani2021analog} introduces analog LCC codes to achieve computations in the analog domain.
Nevertheless, analog LCC codes \cite{soleymani2021analog} do not account for the presence of adversarial workers nor provide perfect privacy guarantees.

In this paper, we focus on enhancing the resiliency of LCC codes against stragglers while providing information-theoretic privacy guarantees for the dataset. We aim to establish a flexible tradeoff among recovery threshold, communication overhead, and computation overhead to optimize system performance. Given the widespread applicability of LCC codes \cite{so2021codedprivateml,so2020scalable,so2022lightsecagg,shao2022dres}, adopting the ideas from the works \cite{soleymani2021list,tang2021verifiable,soleymani2021analog} to enhance the adversarial tolerance of GLCC codes, and extend GLCC codes to operate on the analog domain are interesting avenues for future research.

The rest of this paper is organized as follows. In Section \ref{problem statement}, we formally formulate the problem of distributed multivariate polynomial evaluations. 
In Section \ref{Main:Result}, we present the main results of this paper and discuss its connection to related works. 
In Section \ref{section:GLCC}, we describe the proposed GLCC codes, and analyze their privacy and system performance. 
In Section \ref{experiments}, we apply GLCC to distributed learning, and provide experimental results to demonstrate the performance of GLCC codes. Finally, the paper is concluded in Section \ref{conclusion}.

\subsubsection*{Notation}For a finite set $\mathcal{K}$, $|\mathcal{K}|$ denotes its cardinality.
For any positive integers $m,n$ such that $m<n$, $[n]$ and $[m:n]$ denote the sets $\{1,2,\ldots,n\}$ and $\{m,m+1,\ldots,n\}$, respectively.
Define $Y_{\mathcal{K}}$ as $\{Y_{k_1},Y_{k_2}\ldots,Y_{k_{m}}\}$ for any index set $\mathcal{K}=\{k_1,k_2,\ldots,k_{m}\}\subseteq[n]$.

\section{Problem Formulation}\label{problem statement}
Consider the problem of evaluating a multivariate polynomial $\phi:\mathbb{F}_q^{U}\rightarrow\mathbb{F}_q^{V}$ of total degree at most $D$ over a dataset $X=(X_1,X_2,\ldots,X_M)$ of $M$ input data, where $U$ and $V$ are the  input and output dimensions of polynomial $\phi$, respectively, over some finite field $\mathbb{F}_q$ of size $q$.
We are interested in computing polynomial evaluations of the dataset, over a distributed computing system with a master node and $N$ worker nodes, in which the goal of the master is to compute $Y=(Y_1,Y_2,\ldots,Y_M)$ such that $Y_m\triangleq \phi(X_m)$ for all $m\in[M]$ in the presence of up to $A$ adversarial workers that 
return arbitrarily erroneous computation results,
while keeping the dataset private from \emph{any} colluding subset of up to $T$ workers, as shown in Fig. \ref{GLCC}.



For this purpose, a distributed computing protocol operates in three phases: sharing, computation, and reconstruction. The details of these phases are described as follows.
\begin{itemize}
  \item \textbf{Sharing:} In order to keep the dataset private and exploit the computational power at the workers, the master sends a privately encoded version of input data to each worker. The encoded data for worker $n$ is denoted by $\widetilde{X}_{n}$,
  which is generated by computing some encoding function over the dataset $X$ and some random data $Z$ generated privately at the master.
  \item \textbf{Computation:} Upon receiving the encoded data  $\widetilde{X}_{n}$, worker $n$  generates a response $\widetilde{Y}_n$ 
according to the received data $\widetilde{X}_{n}$ and the public polynomial $\phi$.
Then the worker $n$ sends $\widetilde{Y}_{n}$ back to the master.

  \item \textbf{Reconstruction:} 
Due to the heterogeneity of computing resources and unreliable network conditions in distributed computing systems, there are some straggler workers who may fail to respond in time, which prolongs task completion time. In order to speed up computation, the master only waits for the responses from a subset of the fastest $K$ workers for some design parameter $K\leq N$, and then decodes the desired evaluations from their responses. This allows the computation protocol to tolerate any subset of up to $N-K$ stragglers.
\end{itemize}

A distributed computing protocol needs to satisfy the following two basic requirements.
\begin{itemize}
\item\textbf{Privacy Constraint:} Any $T$ colluding workers must not reveal any information about the dataset, i.e.,
\begin{IEEEeqnarray}{c}
I(X;\widetilde{X}_{\mathcal{T}})=0, \quad\forall\, \mathcal{T}\subseteq[N],|\mathcal{T}|=T. \label{LCC:privacy}
\end{IEEEeqnarray}
\item\textbf{Correctness Constraint:} The master must be able to correctly decode the desired evaluations $Y=(Y_1,Y_2,\ldots,Y_M)$ from the collection of responses of \emph{any} fastest $K$ workers, even when up to $A$ out of the $K$ workers are adversarial, i.e.,
\begin{IEEEeqnarray*}{c}
       H(Y|\widetilde{Y}_{\mathcal{K}})=0,\quad\forall\,\mathcal{K}\subseteq[N],|\mathcal{K}|=K. \label{LCC:correctness}
       \end{IEEEeqnarray*}
\end{itemize}

The performance of the distributed computing protocol is measured by the following key quantities:
\begin{enumerate}
  \item[1.] The \emph{recovery threshold} $K$, which is the minimum number of workers the master needs to wait for in order to recover the desired function evaluations.
  \item[2.] The \emph{communication cost}, which is comprised of the normalized upload cost for the dataset and the normalized download cost from the workers, defined as
  \begin{IEEEeqnarray}{c}
  P_u\!\triangleq\!\frac{\sum_{i=1}^{N}\!H(\widetilde{X}_i)}{U},\;\; P_d\!\triangleq\!\max\limits_{\mathcal{K}:\mathcal{K}\subseteq[N],|\mathcal{K}|=K}\frac{H(\widetilde{Y}_{\mathcal{K}})}{V}, \label{upload and download} \IEEEeqnarraynumspace
  \end{IEEEeqnarray}
  where the information entropy is defined in base $q$ units. 
  \item[3.] The \emph{computation complexity}, which includes the complexities of encoding, worker computation, and decoding. The encoding complexity ${C}_{e}$ at the master is defined as the number of arithmetic operations required to generate the encoded data $\widetilde{X}_{[N]}$, normalized by $U$. The complexity of worker computation ${C}_{w}$ is defined as the number of arithmetic operations required to compute the response $\widetilde{Y}_n$, maximized over $n\in[N]$ and normalized by the complexity of evaluating the polynomial on a single input. Finally, the decoding complexity ${C}_d$ at the master is defined as the number of arithmetic operations required to decode the desired evaluations $Y$ from the responses of fastest workers in $\mathcal{K}$, maximized over all $\mathcal{K}\subseteq[N]$ with $|\mathcal{K}|=K$ and normalized by $V$.
\end{enumerate}


Given the  distributed computing framework above, our goal is to design encoding, computing and decoding functions that establish a flexible tradeoff among the recovery threshold, the communication cost and the computation complexity to simultaneously provide resiliency against stragglers, security against adversarial workers and information-theoretic privacy of the dataset.

\section{Main Results and Discussions}\label{Main:Result}
We propose a family of novel Generalized Lagrange Coded Computing (GLCC) codes for the problem of evaluating multivariate polynomial in the distributed computing system.
We state our main results in the following theorem and discuss their connection to related works.

\begin{Theorem}\label{GLCC:theorem}
For computing any multivariate polynomial $\phi$ of total degree at most $D$ over $M$ input data on a finite field $\mathbb{F}_q$ of size $q$, over a distributed computing system of $N$ workers with $T$-colluding privacy constraint and up to $A$ adversary workers, the following performance metrics are achievable as long as $N\geq \big\lceil D(\frac{M-G}{GL}+T)+\frac{(G-1)M+G}{GL}+2A\big\rceil$ and $q\geq M+LN$, for any positive integers $G,L$ such that $G|M$. 
\begin{IEEEeqnarray*}{l}
\text{Recovery Threshold:}\;\;\notag\\
\quad\quad \;
K=\left\lceil D\Big(\frac{M-G}{GL}+T\Big)+\frac{(G-1)M+G}{GL}+2A\right\rceil, \\
\text{Upload Cost:}\, P_u=GLN,\\
\text{Download Cost:}\, P_d=KL,\\
\text{Encoding Complexity:}\, {C}_e\!=\!\mathcal{O}\!\left(\!GNL(\log(\!N\!L))^2\log\log(NL)\!\right)\!, \\
\text{Worker Computation Complexity:}\, {C}_w=\mathcal{O}(GL), \\
\text{Decoding Complexity:}\, {C}_d\!=\!\mathcal{O}(NL(\log (NL))^2\log\log (NL)).
\end{IEEEeqnarray*}
\end{Theorem}

\begin{proof}
Theorem \ref{GLCC:theorem} is formally proven in Section \ref{section:GLCC} by describing and analyzing the proposed GLCC codes.
\end{proof}

\begin{figure*}[htbp]
\centering
\includegraphics[width=15.0cm]{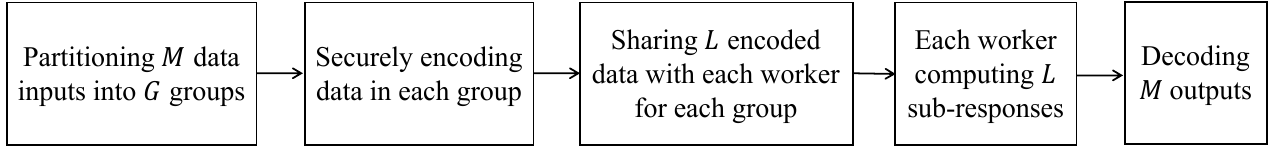}
    \caption{GLCC workflow.}
    \label{GLCC:Flowchart}
\end{figure*}

For the proposed GLCC codes, at a high level, as shown in Fig. \ref{GLCC:Flowchart}, the master first partitions the dataset of size $M$ into $G$ disjoint groups, each containing $M/G$ data inputs. To keep the dataset private from any $T$ colluding workers, the master encodes the $M/G$ data in each of the groups along with $LT$ random noises using a Lagrange interpolating polynomial of degree $M/G+LT-1$, and then shares the evaluations of the polynomial at $L$ distinct points with each worker for each group of inputs. 
This creates computational redundancy across the distributed workers that are used for providing robustness against stragglers, adversaries, and data privacy.
After receiving the privately encoded data for the $G$ groups, each worker first evaluates $\phi$ over the encoded data in each group, which can be viewed as evaluating the composition of a Lagrange polynomial of degree $M/G+LT-1$ with the multivariate polynomial of total degree $D$. The worker then generates $L$ sub-responses by further combining the evaluation results from all groups, using coefficients that are evaluations of carefully chosen polynomials of degree $M-M/G$. Overall, each of the sub-responses at workers can be viewed as an evaluation of a composition polynomial of degree $D\left(M/G+LT-1\right)+\left(M-{M}/{G}\right)$. The master can interpolate this polynomial from any $D\left(M/G+LT-1\right)+\left(M-{M}/{G}\right)+2AL+1$ sub-responses in the presence of $A$ adversarial workers who may provide $AL$ arbitrarily erroneous sub-responses at most, and finally completes the computation by evaluating the polynomial. Since each worker generates $L$ sub-responses, it is enough to wait for the responses from any fastest $K=\big\lceil\frac{D\left(M/G+LT-1\right)+\left(M-{M}/{G}\right)+2AL+1}{L}\big\rceil=\big\lceil D(\frac{M-G}{GL}+T)+\frac{(G-1)M+G}{GL}+2A\big\rceil$ workers such that $KL\geq D\left(M/G+LT-1\right)+\left(M-{M}/{G}\right)+2AL+1$.

\begin{Remark}\label{GLCC:remark}
Our proposed GLCC codes include LCC codes~\cite{LCC} as a special case by setting $G=L=1$. 
By adjusting the parameters $G$ and $L$, GLCC codes achieve a more flexible tradeoff among recovery threshold, communication cost, and computation complexity. It can be clearly observed from Theorem \ref{GLCC:theorem} that, 
as the number of partitioning groups $G$ increases, GLCC codes can reduce the recovery threshold and download cost by increasing the upload cost, encoding complexity and worker computation complexity. Similarly, as the number of sub-responses $L$ at each worker increases, GLCC codes can reduce the recovery threshold by increasing communication and computation costs.
One can optimize the design parameters $G$ and $L$ to minimize the delay of the overall distributed computation. 
Notably, in comparison to our ISIT version \cite{zhu2022generalized}, we allow each worker to generate multiple sub-responses and introduce the additional optimization parameter $L$, further reducing the recovery threshold and achieving a more flexible tradeoff among system performance. Additionally, we include extensive experiments in Section \ref{expe:results} to demonstrate  the performance gain of GLCC codes over LCC codes in a distributed training application.
\end{Remark}

\begin{Remark}\label{field:remark}
The proposed GLCC codes can operate on any finite field $\mathbb{F}_q$ of size $q\geq M+LN$. Apparently, the required finite field size is independent of the parameter $G$ and linearly dependent on the parameter $L$. When a small finite field is desired, the parameter $L$ can be set to $1$, achieving the same finite field size $M+N$ as LCC codes. However, GLCC codes still retain the ability to achieve a lower recovery threshold by adjusting the parameter $G$.
\end{Remark}

\begin{Remark}
As mentioned in the introduction, the problem of evaluating multivariate polynomial of total degree $D$ can be regarded as computing a linear combination of various evaluations of a multi-linear map function, where a multi-linear map function $f$ with $D$ variables $(x_1,x_2,\ldots,x_D)$ is a function satisfying $f(x_1,\!\ldots\!,x_{s-1},ax_s+bx_s',x_{s+1},\!\ldots\!,x_D)\allowbreak=af(x_1,\ldots,x_{s-1},x_s,x_{s+1},\ldots,x_D)+bf(x_1,\ldots,x_{s-1},x_s',x_{s+1},\ldots,x_D)$ for all $s\in[D]$ and $a,b\in \mathbb{F}_q$. 
The problem of evaluating multivariate polynomial was also focused on \cite{jia2021cross} in the special case of $T=A=0$, by presenting multilinear map codes for the problem of evaluating multilinear map function over a batch of data. The multilinear map codes are constructed using the idea of cross subspace alignment that  facilitates an interference alignment structure with a form of Cauchy-Vandermonde matrices, such that the desired computations occupy dimensions corresponding to the Cauchy terms, and the remaining interferences are aligned  within the higher-order terms that constitute the Vandermonde part. 
The computation codes  \cite{jia2021cross} for evaluating multivariate polynomial are obtained by a linear combination of the  multilinear map codes, which achieve identical performance to GLCC codes in the special case of $L=1$ and $T=A=0$.
In contrast to the multi-linear map codes \cite{jia2021cross}, GLCC codes employ Lagrange interpolating polynomial to encode each group of inputs, and then combine the computation results at the $G$ groups for $L$ times using carefully chosen polynomials as coefficients. This achieves a distinct interference alignment structure, such that the master
can interpolate a polynomial from worker responses and finally completes the computation by evaluating the polynomial.
Moreover, allowing each worker to compute and generate $L$ sub-responses further reduces the number of workers the master needs to wait for since the master can obtain more evaluations of the polynomial from each worker. 
\end{Remark}

\begin{Remark}\label{application:SBMM}
GLCC codes can be directly applied to solving specific batch processing problems in coded distributed computing, exemplified by Secure Bath Matrix Multiplication (SBMM). In SBMM, the master aims to compute the pairwise products $(\mathbf{A}^1\mathbf{B}^1,\mathbf{A}^2\mathbf{B}^2,\ldots,\mathbf{A}^{\lambda}\mathbf{B}^{\lambda})$ of two batches of matrices $\mathbf{A}=(\mathbf{A}^1,\mathbf{A}^2,\ldots,\mathbf{A}^{\lambda})$ and $\mathbf{B}=(\mathbf{B}^1,\mathbf{B}^2,\ldots,\mathbf{B}^{\lambda})$ over the distributed computing system while keeping the data matrices $\mathbf{A}$ and $\mathbf{B}$ secure from any $T$ colluding workers. 
There are two solutions for the SBMM problem within the framework of GLCC codes. The first one involves directly applying GLCC codes to complete the desired SBMM computation by setting $M=\lambda$, $X_{m}=(\mathbf{A}^m,\mathbf{B}^m)$ for all $m\in[M]$, and $\phi=\phi(x,y)=xy$. This results in a recovery threshold of $K=\big\lceil 2(\frac{\lambda-G}{GL}+T)+\frac{\lambda(G-1)+G}{GL}\big\rceil$. Alternatively, to achieve a flexible tradeoff among system performance for the SBMM problem, the second solution allows data matrices to be partitioned into smaller sub-matrices. Assume that the matrices $\mathbf{A}^{m}$ and $\mathbf{B}^m$ are horizontally and vertically partitioned into $s\times t$ and $t\times v$ equal-sized sub-matrices, denoted by $\mathbf{A}^{m}=(\mathbf{A}^{m}_{i,j})_{i\in[s],j\in[t]}$ and  $\mathbf{B}^m=(\mathbf{B}^m_{j,k})_{j\in[t],k\in[v]}$, respectively, for all $m\in[\lambda]$. 
Then bilinear complexity \cite{Smirnov,Strassen} can transform the computation of matrix multiplication $\mathbf{A}^m\mathbf{B}^m$ into computing the pairwise products of two batches of sub-matrices $(\tilde{\mathbf{A}}^{m}_{1},\tilde{\mathbf{A}}^{m}_{2},\ldots,\tilde{\mathbf{A}}^{m}_{R})$ and $(\tilde{\mathbf{B}}^{m}_{1},\tilde{\mathbf{B}}^{m}_{2},\ldots,\tilde{\mathbf{B}}^{m}_{R})$, i.e., the matrix multiplication $\mathbf{A}^m\mathbf{B}^m$ can be completed by computing the pairwise products  $(\tilde{\mathbf{A}}^{m}_{1}\tilde{\mathbf{B}}^{m}_{1},\tilde{\mathbf{A}}^{m}_{2}\tilde{\mathbf{B}}^{m}_{2},\ldots,\tilde{\mathbf{A}}^{m}_{R}\tilde{\mathbf{B}}^{m}_{R})$ for all $m\in[\lambda]$, where the positive integer parameter $R=R(s,t,v)$ denotes the bilinear complexity for multiplying two matrices of sizes $s$-by-$t$ and $t$-by-$v$, and $\tilde{\mathbf{A}}^{m}_{r}$ and $\tilde{\mathbf{B}}^{m}_{r}$ are specific linear functions of $\{\mathbf{A}^{m}_{i,j}:i\in[s],j\in[t]\}$ and $\{\mathbf{B}^m_{j,k}:j\in[t],k\in[v]\}$, respectively, for all $r\in[R]$. Accordingly, GLCC codes can perform the SBMM computation by setting $M=\lambda R$, $X_{(m-1)R+r}=(\tilde{\mathbf{A}}^{m}_{r},\tilde{\mathbf{B}}^{m}_{r})$ for all $r\in[R]$ and $m\in[\lambda]$, and $\phi=\phi(x,y)=xy$, resulting in a recovery threshold of $K=\big\lceil 2(\frac{\lambda R-G}{GL}+T)+\frac{\lambda R(G-1)+G}{GL}\big\rceil$.
\end{Remark}

\section{Generalized Lagrange Coded Computing Codes}\label{section:GLCC}
In this section, we describe the GLCC codes, and analyze their privacy guarantees, communication cost, and computation complexity. This provides the proof for Theorem \ref{GLCC:theorem}.

We start with a simple example to illustrate the key idea behind the proposed GLCC codes.
In this example, to establish the connection between GLCC codes and LCC codes, we first present the construction of LCC codes and then demonstrate how LCC codes are generalized to GLCC codes.

\subsection{Illustrative Example}\label{Example:GLCC}
We consider the function $\phi(X_m)=X_m^2$ with system parameters $N=20,M=4,T=1,A=1$, i.e., the master wishes to compute $X_1^2,X_2^2,X_3^2,X_4^2$ from a distributed system with $20$ workers, while providing data privacy against $1$ worker and computation robustness against $1$ adversarial worker.

Assume that $\{\alpha_n,\alpha_{n,1},\alpha_{n,2}\!:\!n\!\in\![20]\}$ are pairwise distinct elements from $\mathbb{F}_q$ such that $\{\alpha_n,\alpha_{n,1},\alpha_{n,2}\!:\!n\in[20]\}\cap[8]\!\!=\!\!\emptyset$.

\noindent\textbf{LCC Codes:} To provide robustness against stragglers and adversarial workers, LCC codes introduce computational redundancies in the distributed computing system using an encoding polynomial $f(x)$ of degree $4$, given by
\begin{IEEEeqnarray}{rCl}
f(x)&=&X_{1}\cdot\frac{(x-2)(x-3)(x-4)(x-5)}{(1-2)(1-3)(1-4)(1-5)}\notag\\
&&+X_{2}\cdot\frac{(x-1)(x-3)(x-4)(x-5)}{(2-1)(2-3)(2-4)(2-5)}\notag\\
&&+X_{3}\cdot\frac{(x-1)(x-2)(x-4)(x-5)}{(3-1)(3-2)(3-4)(3-5)}\notag\\
&&+X_{4}\cdot\frac{(x-1)(x-2)(x-3)(x-5)}{(4-1)(4-2)(4-3)(4-5)}\notag\\
&&+Z\cdot\frac{(x-1)(x-2)(x-3)(x-4)}{(5-1)(5-2)(5-3)(5-4)}, \notag \IEEEeqnarraynumspace
\end{IEEEeqnarray}
where $Z$ is a random uniform noise and is used for providing privacy guarantees.

The master shares the evaluation $f(\alpha_n)$ with each worker $n,n\in[20]$, who computes the response $\phi(f(\alpha_n))=f^2(\alpha_n)$ and sends it back to the master. Obviously, the response at each worker is an evaluation of the polynomial $f^2(x)$ of degree $8$. Therefore, the computational responses across all the workers form a $(20,9)$ RS codeword. The master can recover the polynomial $f^2(x)$ from the responses collected from any $K=11$ workers, even if there is $1$ adversarial worker. 
Finally, evaluating $f^2(x)$ at $x=1,2,3,4$ will obtain the desired computations $X_1^2,X_2^2,X_3^2,X_4^2$.

\noindent\textbf{GLCC Codes:}
Let us set $G=2$ and $L=2$, meaning that the input data are partitioned into $G=2$ groups and the master shares $L=2$ encoded data with each worker for each group of inputs. Specifically, we first partition the $M=4$ data into $G=2$ groups, each containing $2$ inputs, i.e., $(X_{1},X_{2})$ and $(X_{3},X_{4})$. 
Then, we employ LCC codes to encode each group of input data. Note that since the input data are partitioned into multiple groups and each worker receives multiple encoded data for each group of inputs, it is necessary to choose more random noises to provide privacy protection.
To this end, the master chooses four random noises $Z_{1},Z_{2},Z_{3},Z_{4}$ to mask the input data, and then encodes the two groups of data $(X_{1},X_{2},Z_1,Z_2)$ and $(X_{3},X_{4},Z_3,Z_4)$ using Lagrange interpolating polynomials:
\begin{IEEEeqnarray*}{rCl}
f_1(x)&\!=\!&X_{1}\!\cdot\!\frac{(x\!-\!2)(x\!-\!5)(x\!-\!6)}{(1\!-\!2)(1\!-\!5)(1\!-\!6)}\!+\!X_{2}\!\cdot\!\frac{(x\!-\!1)(x\!-\!5)(x\!-\!6)}{(2\!-\!1)(2\!-\!5)(2\!-\!6)} \notag\\ &&\!+Z_{1}\!\cdot\!\frac{(x\!-\!1)(x\!-\!2)(x\!-\!6)}{(5\!-\!1)(5\!-\!2)(5\!-\!6)}\!+\!Z_{2}\!\cdot\!\frac{(x\!-\!1)(x\!-\!2)(x\!-\!5)}{(6\!-\!1)(6\!-\!2)(6\!-\!5)},\\
f_2(x)&\!=\!&X_{3}\!\cdot\!\frac{(x\!-\!4)(x\!-\!7)(x\!-\!8)}{(3\!-\!4)(3\!-\!7)(3\!-\!8)}\!+\!X_{4}\!\cdot\!\frac{(x\!-\!3)(x\!-\!7)(x\!-\!8)}{(4\!-\!3)(4\!-\!7)(4\!-\!8)}\notag\\
&&\!+Z_{3}\!\cdot\!\frac{(x\!-\!3)(x\!-\!4)(x\!-\!8)}{(7\!-\!3)(7\!-\!4)(7\!-\!8)}\!+\!Z_{4}\!\cdot\!\frac{(x\!-\!3)(x\!-\!4)(x\!-\!7)}{(8\!-\!3)(8\!-\!4)(8\!-\!7)}.
\end{IEEEeqnarray*}
The encoded data $\widetilde{X}_{n}$ sent to worker $n,n\in[20]$ are given by
\begin{IEEEeqnarray}{c}\notag
\widetilde{X}_{n}=\big(f_1(\alpha_{n,1}), f_2(\alpha_{n,1}), f_1(\alpha_{n,2}), f_2(\alpha_{n,2}) \big).
\end{IEEEeqnarray}

Then the worker $n$ computes the evaluations of the function $\phi(\cdot)$ at the received encoded data:
\begin{IEEEeqnarray}{c}\notag
\phi(f_1(\alpha_{n,1})),\; \phi(f_2(\alpha_{n,1})),\; \phi(f_1(\alpha_{n,2})),\; \phi(f_2(\alpha_{n,2})).
\end{IEEEeqnarray}
To reduce the communication overhead from the workers to the master, the worker $n$ combines the evaluation results from the $G=2$ groups for $L=2$ times, generating the two sub-responses $\widetilde{Y}_{n,1},\widetilde{Y}_{n,2}$, given by
\begin{IEEEeqnarray*}{rCl}
\widetilde{Y}_{n,1}&=&\phi(f_1(\alpha_{n,1}))\cdot\underbrace{(\alpha_{n,1}-3)(\alpha_{n,1}-4)}_{\text{Interference\; Elimination}}\notag\\
&&\quad\quad\quad\quad+\phi(f_2(\alpha_{n,1}))\cdot\underbrace{(\alpha_{n,1}-1)(\alpha_{n,1}-2)}_{\text{Interference\; Elimination}},\notag\\
\widetilde{Y}_{n,2}&=&\phi(f_1(\alpha_{n,2}))\cdot\underbrace{(\alpha_{n,2}-3)(\alpha_{n,2}-4)}_{\text{Interference\; Elimination}} \notag\\
&&\quad\quad\quad\quad+\phi(f_2(\alpha_{n,2}))\cdot\underbrace{(\alpha_{n,2}-1)(\alpha_{n,2}-2)}_{\text{Interference\; Elimination}},
\end{IEEEeqnarray*}
where these coefficients are used to eliminate the interference between the evaluation results from the $G=2$ groups.

The sub-responses $\widetilde{Y}_{n,1}$ and $\widetilde{Y}_{n,2}$ are equivalent to the evaluations of the following composition polynomial $h(x)$ of degree $8$ at points $x=\alpha_{n,1}$ and $\alpha_{n,2}$, respectively.
\begin{IEEEeqnarray*}{rCl}
h(x)&\!=\!&\phi(f_1(x))\cdot(x-3)(x-4)\!+\!\phi(f_2(x))\cdot(x-1)(x-2)\\
&\!=\!&(f_1(x))^2\cdot(x-3)(x-4)\!+\!(f_2(x))^2\cdot(x-1)(x-2).
\end{IEEEeqnarray*}

Similar to the decoding process of LCC codes, $(\widetilde{Y}_{1,1},\widetilde{Y}_{1,2},\allowbreak \ldots, \widetilde{Y}_{20,1},\widetilde{Y}_{20,2})=(h(\alpha_{1,1}),h(\alpha_{1,2}),\ldots,\allowbreak h(\alpha_{20,1}),h(\alpha_{20,2}))$ forms a $(40,9)$ RS codeword. Since each worker generates $L=2$ sub-responses, it is enough to decode $h(x)$ from the responses of any $K=7$ out of $20$ workers, in the presence of $1$ adversarial worker who provides at most $2$ erroneous sub-responses. 

Having recovered $h(x)$, the master evaluates it at points $x=1,2,3,4$ to obtain
\begin{IEEEeqnarray*}{c}
h(1)=6 X_{1}^2,\;\; h(2)=2 X_{2}^2,\;\;   h(3)=2 X_{3}^2,\;\; h(4)=6 X_{4}^2.
\end{IEEEeqnarray*}
The master completes the desired evaluations $X_1^2,X_2^2,X_3^2,X_4^2$ by eliminating the interference from constant coefficients.

Privacy is guaranteed because the input data $X_1,X_2$ and $X_3,X_4$ are padded with independently and uniformly random noises $Z_1,Z_2$ and $Z_3,Z_4$, respectively. This ensures that the encoded data $\widetilde{X}_{n}$ received by any $T=1$ worker is independent of the original  input data. 

Although both GLCC codes and LCC codes can accomplish the desired evaluations by decoding a polynomial of degree $8$, GLCC codes allow each worker to generate two evaluations of this polynomial under the privacy constraint. 
Therefore, compared to LCC codes, GLCC codes achieve a lower recovery threshold by requiring responses from fewer workers to decode this polynomial.
Notably, the use of more random noises to provide privacy guarantees and the utilization of interference elimination coefficients to reduce download cost have increased the degree of the decoding polynomial for GLCC codes, making it equal to the degree of the decoding polynomial for LCC codes in this example. However, as the degree of the multivariate polynomial function to be computed increases, GLCC codes can achieve a lower degree of  decoding polynomial than LCC codes due to partitioning the input data into multiple groups.

In general, both GLCC codes and LCC codes exploit the structure inspired by Lagrange interpolating polynomial to encode the input data, such that the responses at workers can be viewed as evaluations of a composition polynomial and accordingly the decoding is completed by interpolating the polynomial. 
In contrast to LCC codes, GLCC codes propose to partition the inputs into $G$ groups and performs Lagrange encoding within each group, and then combines the computation results from the $G$ groups for $L$ times using carefully designed interference eliminating coefficients.
By allowing each worker to process $GL>1$ encoded inputs and generating $L>1$ sub-responses,  GLCC codes include LCC codes as a special case by setting $G=L=1$ and achieve a more flexible tradeoff among recovery threshold, communication cost, and computation complexity.
The key ideas of GLCC codes can be summarized in three aspects:
\begin{itemize}
    \item \textbf{Grouping Strategy:} The grouping strategy reduces the degree $M/G+LT-1$ of the interpolating polynomial in each group of input data, thereby decreasing the degree of the overall polynomial interpolated at the master. This leads to a lower recovery threshold.
\item \textbf{Interference Elimination:} Using interference eliminating coefficients to combine the computation results from the $G$ groups effectively eliminates the interference among all the groups. This guarantees correctness constraint and particularly reduces the download overhead from workers to the master.
\item \textbf{Multiple Sub-Responses:} While generating more than one sub-response ($L>1$) induces a higher degree ($M/G+LT-1$) for the interpolating polynomial in each group, it still reduces the number of workers the master needs to wait for (i.e., the recovery threshold). This is because each worker generates multiple sub-responses, enabling the responses from fewer workers to compensate for the increased degree of the polynomial interpolated at the master.
\end{itemize}


Next, we formally describe the general construction of GLCC codes.

\subsection{General Description of GLCC Codes}
Let $G,L$ be any positive integers such that $G|M$. Then we divide the dataset $X_1,\ldots,X_M$ evenly into $G$ groups, each containing $ R\triangleq M/G$ data inputs. Denote the $r$-th input of group $g$ by $X_{g,r}$, which is given by
\begin{IEEEeqnarray}{c}\label{GLCC:group}
X_{g,r}=X_{(g-1)R+r}, \quad\forall\,g\in[G], r \in[ R].
\end{IEEEeqnarray}

Next, we present a group of elements $\{\beta_{g,r},\alpha_{n,\ell}:g\in[G],r\in[R+LT],n\in[N],\ell\in[L]\}$ from $\mathbb{F}_q$, which are used as the coding parameters of GLCC codes.
To ensure the achievability of GLCC codes, 
the group of elements needs to satisfy the following four constraints.
\begin{enumerate}
  \item[P1.] For each given $g\in[G]$, $\beta_{g,j}\neq\beta_{g,k}$ for all $j,k\in[R+LT]$ with $j\neq k$;
  \item[P2.] The elements $\{\beta_{g,r}:g\in[G],r\in[R]\}$ are pairwise distinct, i.e., $\beta_{g,r}\neq\beta_{g',r'}$ for any given $g,g'\in[G]$ and $r,r'\in[R]$ with $(g,r)\neq (g',r')$;
  \item[P3.] The elements $\{\alpha_{n,\ell}:n\in[N],\ell\in[L]\}$ are all distinct, i.e., $\alpha_{n,\ell}\neq\alpha_{n',\ell'}$ for all $n,n'\in[N]$ and $\ell,\ell'\in[L]$ with $(n,\ell)\neq (n',\ell')$;
  \item[P4.] The elements in $\{\alpha_{n,\ell}:n\in[N],\ell\in[L]\}$ are distinct from the ones in $\{\beta_{g,r}:g\in[G],r\in[R]\}$, i.e., $\{\alpha_{n,\ell}:n\in[N],\ell\in[L]\}\cap\{\beta_{g,r}:g\in[G],r\in[R]\}
  =\emptyset$.
\end{enumerate}

Let $\{\beta_{g,r},\alpha_{n,\ell}:g\in[G],r\in[R],n\in[N],\ell\in[L]\}$ represent any $GR+LN=M+LN$ pairwise distinct elements from $\mathbb{F}_q$, and $\{\beta_{g,r}:r\in[R+1:R+LT]\}$ be a subset of $\{\alpha_{n,\ell}:n\in[N],\ell\in[L]\}$ with size $LT$ for all $g\in[G]$. It is straightforward to prove that these elements satisfy P1-P4 with
$|\{\beta_{g,r},\alpha_{n,\ell}:g\in[G],r\in[R+LT],n\in[N],\ell\in[L]\}|=M+LN$.

We will describe the construction of GLCC codes using the group of elements satisfying P1-P4. This shows that GLCC codes can operate on arbitrary finite filed $\mathbb{F}_q$ of size $q\!\geq\! M\!+\!LN$.


To guarantee the privacy of data inputs in group $g$ for each $g\in[G]$, the master samples independently and uniformly $LT$ random variables $Z_{g,R+1},Z_{g,R+2},\ldots,Z_{g, R+LT}$ from $\mathbb{F}_q^{U}$, and then constructs a polynomial $f_g(x)$ of degree at most $ R+LT-1$ such that
\begin{IEEEeqnarray}{c}\label{encdoing}
f_g(\beta_{g, r })=\left\{
\begin{array}{@{}ll}
X_{g, r },&\forall\,  r \in[ R]\\
Z_{g, r },&\forall\,  r \in[ R+1: R+LT]
\end{array}\right..
\end{IEEEeqnarray}
By P1 and Lagrange interpolating rule, we can exactly express $f_g(x)$ as
\begin{IEEEeqnarray}{l}\label{encdoing function}
f_g(x)=\sum\limits_{j=1}^{ R}X_{g,j}\cdot\prod_{k\in[ R+LT]\backslash\{j\}}\frac{x-\beta_{g,k}}{\beta_{g,j}-\beta_{g,k}} \notag\\
\quad\quad\quad\quad\quad
+\sum\limits_{j=R+1}^{R+LT}Z_{g,j}\cdot\prod_{k\in[ R+LT]\backslash\{j\}}\frac{x-\beta_{g,k}}{\beta_{g,j}-\beta_{g,k}}. \IEEEeqnarraynumspace
\end{IEEEeqnarray}

For each $n\in [N]$, the master computes the evaluation of the polynomial $f_g(x)$ at point $\alpha_{n,\ell}$ for all $g\in [G]$ and $\ell\in[L]$, and sends these evaluation results to worker $n$.  Hence, a total of $GL$ encoded inputs are sent to worker $n$ for all $n\in[N]$, given by
\begin{IEEEeqnarray}{c}\label{GLCC:upload}
\widetilde{X}_n=\big\{f_1(\alpha_{n,\ell}),f_2(\alpha_{n,\ell}),\ldots,f_G(\alpha_{n,\ell}):\ell\in[L]\big\}.
\end{IEEEeqnarray}



The response $\widetilde{Y}_n$ returned by worker $n$ to the master consists of $L$ sub-responses, denoted by
\begin{IEEEeqnarray}{c}\notag
\widetilde{Y}_{n}=\big\{\widetilde{Y}_{n,1},\widetilde{Y}_{n,2},\ldots,\widetilde{Y}_{n,L}\big\},   
\end{IEEEeqnarray}
where the sub-response $\widetilde{Y}_{n,\ell}$ is generated by evaluating the polynomial $\phi$ over the $G$ encoded inputs $f_1(\alpha_{n,\ell}),f_2(\alpha_{n,\ell}),\ldots,f_G(\alpha_{n,\ell})$ and then computing a linear combination of these evaluations for all $\ell\in[L]$, as follow,
\begin{IEEEeqnarray}{c}\label{GLCC:download}
\widetilde{Y}_{n,\ell}=\sum\limits_{g=1}^{G}\left(\phi(f_g(\alpha_{n,\ell}))\cdot\prod\limits_{g'\in[G]\backslash\{g\}}\prod\limits_{r\in[ R]}(\alpha_{n,\ell}-\beta_{g',r})\right).\IEEEeqnarraynumspace
\end{IEEEeqnarray}


Apparently, for any non-adversary worker $n$, the sub-response $\widetilde{Y}_{n,\ell}$ is equivalent to evaluating the following polynomial $h(x)$ at the point $x=\alpha_{n,\ell}$.
\begin{IEEEeqnarray}{c}\label{GLCC:answer fun}
h(x)=\sum\limits_{g=1}^{G}\left(\phi(f_g(x))\cdot\prod\limits_{g'\in[G]\backslash\{g\}}\prod\limits_{ r \in[ R]}(x-\beta_{g', r })\right). \IEEEeqnarraynumspace
\end{IEEEeqnarray}

Since $\phi$ is a multivariate polynomial with total degree at most $D$ and $f_g(x)$ is a polynomial of degree $R+LT-1$ for any $g\in[G]$ by \eqref{encdoing function},
the composite polynomial $\phi(f_g(x))$ has degree at most $D(R+LT-1)$.
Accordingly, $h(x)$ can be viewed as a polynomial in variable $x$ with degree at most $D(R+LT-1)+(G-1)R$.
Recall from P3 that $\{\alpha_{n,\ell}:\ell\in[L],n\in[N]\}$ are pairwise distinct elements in $ \mathbb{F}_q$.
Thus, the sub-responses $(\widetilde{Y}_{1,1},\ldots,\widetilde{Y}_{1,L},\ldots,\widetilde{Y}_{N,1},\ldots,\widetilde{Y}_{N,L})=(h(\alpha_{1,1}),\ldots,h(\alpha_{1,L}),\ldots,h(\alpha_{N,1}),\ldots,h(\alpha_{N,L}))$ forms an $(NL,D(R+LT-1)+(G-1)R+1)$ RS codeword. Notably, there are up to $A$ adversary workers who may provide $AL$ arbitrarily erroneous sub-responses at most. 
By using RS decoding algorithms \cite{Lin,Gao}, the master can decode the polynomial $h(x)$ from at most $D(R+LT-1)+(G-1)R+2AL+1$ sub-responses, in the presence of $A$ adversary workers. Thus the proposed GLCC codes achieve a recovery threshold of $K=\big\lceil D(\frac{R-1}{L}+T)+\frac{(G-1)R+1}{L}+2A\big\rceil$ such that $KL\geq D(R+LT-1)+(G-1)R+2AL+1$.

For any $g\in[G]$ and $r \in[ R]$, evaluating $h(x)$ at $x=\beta_{g, r }$ yields
\begin{IEEEeqnarray}{rCl}\label{recover:result}
h(\beta_{g, r })&=&\sum\limits_{j=1}^{G}\Big(\phi(f_j(\beta_{g, r }))\cdot\prod\limits_{g'\in[G]\backslash\{j\}}\prod\limits_{k\in[ R]}(\beta_{g, r }-\beta_{g',k})\Big)\notag\\
&\overset{(a)}{=}&\phi(f_g(\beta_{g, r }))\cdot\prod\limits_{g'\in[G]\backslash\{g\}}\prod\limits_{k\in[ R]}(\beta_{g, r }-\beta_{g',k})\notag\\
&\overset{(b)}{=}&\phi(X_{g, r })\cdot\prod\limits_{g'\in[G]\backslash\{g\}}\prod\limits_{k\in[ R]}(\beta_{g, r }-\beta_{g',k}),\label{GLCC:evaluations}
\end{IEEEeqnarray}
where $(a)$ is due to $\prod_{g'\in[G]\backslash\{j\}}\prod_{k\in[ R]}(\beta_{g, r }-\beta_{g',k})=0$ for all $j\in[G]\backslash\{g\}$, and $(b)$ follows by \eqref{encdoing}.
From P2 and \eqref{GLCC:group}, we can obtain the evaluation $Y_{(g-1) R+ r }=\phi(X_{(g-1) R+ r })=\phi(X_{g, r })$ by eliminating the non-zero constant term $\prod_{g'\in[G]\backslash\{g\}}\prod_{k\in[ R]}(\beta_{g, r }-\beta_{g',k})$ in $h(\beta_{g, r })$.
Finally, the master completes the desired computation $Y=(Y_1,Y_2,\ldots,Y_M)$ after traversing all $g\in[G], r \in[ R]$.

As a result, for any given parameters $G,L$ such that  $G|M$, the proposed GLCC codes achieve the recovery threshold $K=\big\lceil D(\frac{R-1}{L}+T)+\frac{(G-1)R+1}{L}+2A\big\rceil$ for any finite field of size $q\geq M+LN$.



\begin{Remark}\label{GLCC:pstragglers}
While GLCC codes treat stragglers as erasures and discard the computations done by the straggler workers, 
they can also exploit partial computation results at straggler workers to further accelerate the computation.
Specifically, in our GLCC codes, the $L$ sub-responses $\widetilde{Y}_{n,1},\widetilde{Y}_{n,2},\ldots,\widetilde{Y}_{n,L}$ \eqref{GLCC:download} at each worker $n,n\in[N]$ can be computed sequentially, and the worker $n$ sends the sub-response $\widetilde{Y}_{n,\ell}$ back to the master as soon as it is completed for any $\ell\in[L]$.
The master can recover the polynomial $h(x)$ \eqref{GLCC:answer fun} and then complete desired computation once it receives any $D(R+LT-1)+(G-1)R+2AL+1$ sub-responses from the $N$ workers. For the example in Section \ref{Example:GLCC}, it is sufficient to just wait for any $13$ out of $20$ sub-responses, as highlighted in red in Table \ref{tab:GLCC}.  
However, how to design coding schemes which can efficiently exploit the partial computation results done by slower straggler workers is another interesting topic worthy of specialized research \cite{ozfatura2021coded,das2022coded,kianidehkordi2020hierarchical}. In this paper, we only consider the case where a non-straggling worker is available to generate all the sub-responses, and a straggler is not able to provide any computation results. 
\end{Remark}

\begin{table*}[htbp]
\centering
\caption{An example of GLCC codes with partial stragglers. All the parameters are identical to those in Section \ref{Example:GLCC}. The desired computation can be completed with any $13$ out of $20$ sub-responses, such as those highlighted in red.}\label{tab:GLCC}
  \begin{tabular}{|c|c|c|c|c|c|}
  \hline
  & Worker $1$ & Worker $2$ & Worker $3$ & Worker $4$ & Worker $5$ \\ \hline
\multirow{2}{*}{Coded Storage} & $f_1(\alpha_{1,1})\;\, f_2(\alpha_{1,1})$ & $f_1(\alpha_{2,1})\;\, f_2(\alpha_{2,1})$ & $f_1(\alpha_{3,1})\;\, f_2(\alpha_{3,1})$ & $f_1(\alpha_{4,1})\;\, f_2(\alpha_{4,1})$ & $f_1(\alpha_{5,1})\;\, f_2(\alpha_{5,1})$ \\ \Xcline{2-6}{0.5pt}
& $f_1(\alpha_{1,2})\;\, f_2(\alpha_{1,2})$ & $f_1(\alpha_{2,2})\;\, f_2(\alpha_{2,2})$ & $f_1(\alpha_{3,2})\;\, f_2(\alpha_{3,2})$ & $f_1(\alpha_{4,2})\;\, f_2(\alpha_{4,2})$ & $f_1(\alpha_{5,2})\;\, f_2(\alpha_{5,2})$ \\
\hline
Sub-response 1 & \red{$\widetilde{Y}_{1,1}$} & \red{$\widetilde{Y}_{2,1}$} & \red{$\widetilde{Y}_{3,1}$} & \red{$\widetilde{Y}_{4,1}$} & \red{$\widetilde{Y}_{5,1}$} \\ \hline
Sub-response 2 & \red{$\widetilde{Y}_{1,2}$} & \red{$\widetilde{Y}_{2,2}$} & $\widetilde{Y}_{3,2}$ & $\widetilde{Y}_{4,2}$ & \red{$\widetilde{Y}_{5,2}$} \\ \hline\hline
& Worker $6$ & Worker $7$ & Worker $8$ & Worker $9$ & Worker $10$ \\ \hline
\multirow{2}{*}{Storage} & $f_1(\alpha_{6,1})\;\, f_2(\alpha_{6,1})$ & $f_1(\alpha_{7,1})\;\, f_2(\alpha_{7,1})$ & $f_1(\alpha_{8,1})\;\, f_2(\alpha_{8,1})$ & $f_1(\alpha_{9,1})\;\, f_2(\alpha_{9,1})$ & $f_1(\alpha_{10,1})\;\, f_2(\alpha_{10,1})$ \\ \Xcline{2-6}{0.5pt}
& $f_1(\alpha_{6,2})\;\, f_2(\alpha_{6,2})$ & $f_1(\alpha_{7,2})\;\, f_2(\alpha_{7,2})$ & $f_1(\alpha_{8,2})\;\, f_2(\alpha_{8,2})$ & $f_1(\alpha_{9,2})\;\, f_2(\alpha_{9,2})$ & $f_1(\alpha_{10,2})\;\, f_2(\alpha_{10,2})$ \\ \hline
Sub-response 1 & \red{$\widetilde{Y}_{6,1}$} & \red{$\widetilde{Y}_{7,1}$} & \red{$\widetilde{Y}_{8,1}$} & \red{$\widetilde{Y}_{9,1}$} & $\widetilde{Y}_{10,1}$ \\ \hline
Sub-response 2 & $\widetilde{Y}_{6,2}$ & $\widetilde{Y}_{7,2}$ & \red{$\widetilde{Y}_{8,2}$} & $\widetilde{Y}_{9,2}$ & $\widetilde{Y}_{10,2}$ \\ \hline
  \end{tabular}
\end{table*}


\subsection{Privacy, Communication and Computation Overheads}
In this subsection, we prove the privacy of the proposed GLCC codes and analyze their communication cost and computation complexity. 

\begin{Lemma}[Generalized Cauchy Matrix \cite{Lin}]\label{g-cauchy matrix}
Let $\alpha_1,\!\ldots\!,\alpha_{\tau}$ and $\beta_1,\!\ldots\!,\beta_{\tau}$ be pairwise distinct elements from $\mathbb{F}_q$ for any positive integer $\tau$, and $v_1,\!\ldots\!,v_{\tau},s_1,\!\ldots\!,s_{\tau}$ be $2\tau$ nonzero elements from $\mathbb{F}_q$. Denote by $f_k(x)$ a polynomial of degree $\tau\!-\!1$
\begin{IEEEeqnarray}{c}\notag
f_k(x)=\prod\limits_{j\in[\tau]\backslash\{k\}}\frac{x-\beta_j}{\beta_k-\beta_j},\quad\forall\,k\in[\tau].
\end{IEEEeqnarray}
Then the following generalized Cauchy matrix is invertible over $\mathbb{F}_q$.
\begin{IEEEeqnarray}{c}\notag
\left[
  \begin{array}{@{\,}cccc@{\,}}
    v_1 f_{1}(\alpha_1)s_1 & v_1 f_{2}(\alpha_1)s_2& \ldots & v_{1} f_{\tau}(\alpha_1)s_{\tau}  \\
    v_2 f_{1}(\alpha_2)s_1 & v_2 f_{2}(\alpha_2)s_2& \ldots & v_{2} f_{\tau}(\alpha_2)s_{\tau}  \\
    \vdots & \vdots & \ddots & \vdots \\
    v_{\tau} f_{1}(\alpha_{\tau})s_1 & v_{\tau} f_{2}(\alpha_{\tau})s_2& \ldots & v_{\tau} f_{\tau}(\alpha_{\tau})s_{\tau}  \\
\end{array}
\right]_{\tau\times \tau}.
\end{IEEEeqnarray}
\end{Lemma}


\begin{Lemma}[Corollaries 10.8 and 10.12 in \cite{Von}]\label{Lamma:complexity11}
The evaluation of a $k$-th degree polynomial at $k+1$ arbitrary points can be done in ${\mathcal{O}}(k(\log k)^2\log\log k)$ arithmetic operations, and consequently, its dual problem, interpolation of a $k$-th degree polynomial from $k+1$ arbitrary points can also be performed in ${\mathcal{O}}(k(\log k)^2\log\log k)$ arithmetic operations.
\end{Lemma}


\begin{Lemma}[Decoding Reed-Solomon Codes \cite{Gao,chen2008complexity,von2013modern,van2022optimizing}]\label{Lamma:complexity33}
Decoding Reed-Solomon codes of dimension $n$ with $a$ erroneous symbols over arbitrary finite fields can be done in $\mathcal{O}(n(\log n)^2\log\log n)$ arithmetic operations if its minimum distance $d$ satisfies $d>2a$.\footnote{The Half Greatest Common Divisor (Half-GCD) algorithm \cite{van2022optimizing} serves as a fast and efficient method for decoding RS codes. While RS codes can achieve a decoding complexity of $\mathcal{O}(n(\log n)^2\log\log n)$ for arbitrary finite fields by applying the Half-GCD algorithm outlined in Algorithm 11.6 in \cite{von2013modern} to the RS decoding algorithm in \cite{Gao}, it introduces a relatively large constant factor of 22. Depending on specific decoding scenarios, such as whether the finite field used supports Fast Fourier Transform (FFT), it is advisable to choose an appropriate Half-GCD algorithm  \cite{van2022optimizing} to mitigate the impact of the constant factor.
}
\end{Lemma}

\subsubsection*{Privacy} We demonstrate that GLCC codes provide a privacy guarantee of data inputs against any $T$ colluding workers. At a high level, the data inputs of each group are protected using $LT$ random noises, ensuring that the $LT$ encoded data received by any $T$ workers does not disclose any information about the group of data inputs. 
Moreover, since the random noises used to mask data inputs are independent across all the groups, the combination of the encoded data obtained from all the groups remains unknown to all data inputs. Notably, each worker is shared with $L$ encoded data for the data inputs of each group, and thus GLCC utilizes $LT$ random noises to protect data privacy against $T$ colluding workers. Once the number of random noises used is reduced, then any $T$ colluding workers will acquire at least a linear combination of the group of data inputs from the received $LT$ encoding data by \eqref{encdoing function}, thereby violating privacy constraint.

We next present a formal description. 
Let $\mathcal{T}\!=\!\{n_1,\!\ldots\!,n_{T}\}$ be any $T$ indices out of the $N$ workers.
From \eqref{encdoing function} and \eqref{GLCC:upload}, the $LT$ encoded data $\{f_g(\alpha_{n,\ell}):n\in\mathcal{T},\ell\in[L]\}$ sent to the workers
in $\mathcal{T}$ are protected by $LT$ random noises $Z_{g,R+1},\ldots,Z_{g, R+LT}$ for each $g\in[G]$, given by 
\begin{IEEEeqnarray}{l}
\left[
\begin{array}{@{}c@{}}
f_g(\alpha_{n_1,1}) \\ \vdots \\ f_g(\alpha_{n_1,L}) \\ \vdots \\ f_g(\alpha_{n_T,1}) \\ \vdots \\ f_g(\alpha_{n_T,L})
\end{array}
\right]=\underbrace{
\left[
\begin{array}{@{}c@{}}
u_g(\alpha_{n_1,1}) \\ \vdots \\ u_g(\alpha_{n_1,L}) \\ \vdots \\ u_g(\alpha_{n_T,1}) \\ \vdots \\ u_g(\alpha_{n_T,L})
\end{array}
\right]}_{=\mathbf{u}_g^{\mathcal{T}}} + \notag\\
\underbrace{
\left[
\begin{array}{@{}ccc@{}}
  b_{g,R+1}(\alpha_{n_1,1}) & \ldots & b_{g,R+LT}(\alpha_{n_1,1}) \\
  \vdots & \vdots & \vdots \\
b_{g,R+1}(\alpha_{n_1,L}) & \ldots & b_{g,R+LT}(\alpha_{n_1,L}) \\
\vdots & \vdots & \vdots \\
b_{g,R+1}(\alpha_{n_T,1}) & \ldots & b_{g,R+LT}(\alpha_{n_T,1}) \\
\vdots & \vdots & \vdots \\
b_{g,R+1}(\alpha_{n_T,L}) & \ldots & b_{g,R+LT}(\alpha_{n_T,L}) \\
\end{array}
\right]}_{=\mathbf{F}_{g}^{\mathcal{T}}}
\underbrace{
\left[
\begin{array}{@{}c@{}}
  Z_{g,R+1} \\
  \vdots \\
  Z_{g,R+LT}
\end{array}
\right]}_{=\mathbf{z}_g}, \IEEEeqnarraynumspace \label{answers:form}
\end{IEEEeqnarray}
where for all $g\in[G]$,
\begin{IEEEeqnarray}{c}\label{security:proof2}
u_g(x)\!=\!\sum_{j=1}^{ R}X_{g,j}\cdot\prod_{k\in[R+LT]\backslash\{j\}}\frac{x-\beta_{g,k}}{\beta_{g,j}-\beta_{g,k}}, \IEEEeqnarraynumspace
\end{IEEEeqnarray}
and for all $j\in[R+1:R+LT]$,
\begin{IEEEeqnarray}{c}\notag
b_{g,j}(x)=\prod_{k\in[R]}\frac{x-\beta_{g,k}}{\beta_{g,j}-\beta_{g,k}}\cdot\prod_{k\in[R+1:R+LT]\backslash\{j\}}\frac{x-\beta_{g,k}}{\beta_{g,j}-\beta_{g,k}}. \label{security:proof}
\end{IEEEeqnarray}

It is easy to prove that the matrix $\mathbf{F}_{g}^{\mathcal{T}}$ is non-singular over $ \mathbb{F}_q$ for any $\mathcal{T}$ and $g\in[G]$ by Lemma \ref{g-cauchy matrix} and P1-P4. Let $(\mathbf{F}_{g}^{\mathcal{T}})^{-1}$ denote the inverse matrix of $\mathbf{F}_{g}^{\mathcal{T}}$, then we have 
\begin{IEEEeqnarray*}{rCl}
&&I(X;\widetilde{X}_{\mathcal{T}})\notag\\
&\overset{(a)}{=}&I(\{X_{g,1},\!\ldots\!,X_{g, R}\}_{g\in[G]};\{f_g(\alpha_{n,\ell}):n\!\in\!\mathcal{T},\ell\!\in\![L]\}_{g\in[G]})\\
&\overset{(b)}{=}&I(\{X_{g,1},\ldots,X_{g,R}\}_{g\in[G]};\{\mathbf{u}_g^{\mathcal{T}}+\mathbf{F}_g^{\mathcal{T}}\cdot\mathbf{z}_g\}_{g\in[G]})\\
&=&I(\{X_{g,1},\ldots,X_{g,R}\}_{g\in[G]};\{(\mathbf{F}_{g}^{\mathcal{T}})^{-1}\cdot\mathbf{u}_g^{\mathcal{T}}+\mathbf{z}_g\}_{g\in[G]})\\
&=&H(\{(\mathbf{F}_{g}^{\mathcal{T}})^{-1}\cdot\mathbf{u}_g^{\mathcal{T}}+\mathbf{z}_g\}_{g\in[G]})\notag\\
&&\;\;-H(\{(\mathbf{F}_{g}^{\mathcal{T}})^{-1}\cdot\mathbf{u}_g^{\mathcal{T}}+\mathbf{z}_g\}_{g\in[G]}|\{X_{g,1},\ldots,X_{g,R}\}_{g\in[G]})\\
&\overset{(c)}{=}&H(\{(\mathbf{F}_{g}^{\mathcal{T}})^{-1}\cdot\mathbf{u}_g^{\mathcal{T}}+\mathbf{z}_g\}_{g\in[G]})\notag\\
&&\quad\quad\quad\quad\quad-H(\{\mathbf{z}_g\}_{g\in[G]}|\{X_{g,1},\ldots,X_{g,R}\}_{g\in[G]})\\
&\overset{(d)}{=}&H(\{(\mathbf{F}_{g}^{\mathcal{T}})^{-1}\cdot\mathbf{u}_g^{\mathcal{T}}+\mathbf{z}_g\}_{g\in[G]})-H(\{\mathbf{z}_g\}_{g\in[G]})\\
&\overset{(e)}{=}&0,
\end{IEEEeqnarray*}
where $(a)$ is due to \eqref{GLCC:group} and \eqref{GLCC:upload}; 
$(b)$ follows by \eqref{answers:form};
$(c)$ holds because $\{(\mathbf{F}_{g}^{\mathcal{T}})^{-1}\cdot\mathbf{u}_g^{\mathcal{T}}\}_{g\in[G]}$ is a deterministic function of $\{X_{g,1},\ldots,X_{g,R}\}_{g\in[G]}$ by \eqref{answers:form}-\eqref{security:proof2};
$(d)$ is because $\{\mathbf{z}_g\}_{g\in[G]}$ are generated independently of $\{X_{g,1},\ldots,X_{g,R}\}_{g\in[G]}$;
$(e)$ is due to the fact that all the noises $\{Z_{g,R+1},\ldots,Z_{g,R+LT}\}_{g\in[G]}$ in $\{\mathbf{z}_g\}_{g\in[G]}$ are i.i.d. uniformly distributed on $\mathbb{F}_q^U$, and are independent of $\{(\mathbf{F}_{g}^{\mathcal{T}})^{-1}\cdot\mathbf{u}_g^{\mathcal{T}}\}_{g\in[G]}$, such that $\{(\mathbf{F}_{g}^{\mathcal{T}})^{-1}\cdot\mathbf{u}_g^{\mathcal{T}}+\mathbf{z}_g\}_{g\in[G]}$ and $\{\mathbf{z}_g\}_{g\in[G]}$ are identically and uniformly distributed over $\mathbb{F}_q^{LTGU}$, i.e., $H(\{(\mathbf{F}_{g}^{\mathcal{T}})^{-1}\cdot\mathbf{u}_g^{\mathcal{T}}+\mathbf{z}_g\}_{g\in[G]})=H(\{\mathbf{z}_g\}_{g\in[G]})=LTGU$.

This proves that the GLCC codes satisfy the privacy constraint in \eqref{LCC:privacy}.

\subsubsection*{Communication Cost}
The master shares $GL$ encoded inputs of each size $U$ to each worker by \eqref{GLCC:upload}, and downloads $L$ sub-responses of each size $V$ from each of responsive workers by \eqref{GLCC:download}.  Thus, from \eqref{upload and download}, the normalized upload  and download cost are given by $P_u=GLUN/U=GLN$ and $P_d=KLV/V=KL$, respectively.

\subsubsection*{Computation Complexity}The encoding process for the dataset can be viewed as evaluating $G$ polynomials of degree $ R+LT-1<NL$ at $NL$ points for $U$ times by \eqref{encdoing function} and \eqref{GLCC:upload}, which incurs the normalized complexity of $\mathcal{O}(GNL(\log(NL))^2\log\log(NL))$ by Lemma \ref{Lamma:complexity11}.
In terms of computation complexity at each worker, to generate a sub-response \eqref{GLCC:download}, the computation mainly includes evaluating the multivariate polynomial $\phi$ over $G$ data inputs, which incurs a normalized complexity of $G$.  Note that the terms $\prod_{g'\in[G]\backslash\{g\}}\prod_{r\in[R]}(\alpha_{n,\ell}-\beta_{g',r}),g\in[G]$ in \eqref{GLCC:download} are
independent of the dataset and thus can be computed at worker $n$ a priori to reduce the latency of worker computation, and thus its complexity is negligible. Since each worker generates $L$ sub-responses, the computation complexity at each worker is $\mathcal{O}(GL)$.
For decoding complexity, the master decodes a polynomial \eqref{GLCC:answer fun} of degree less than $KL$ from an $NL$-dimensional RS codeword with at most $AL$ errors and then evaluates the polynomial at $GR=M<KL$ points \eqref{GLCC:evaluations}, for $V$ times. This achieves the normalized complexity of $\mathcal{O}(NL(\log (NL))^2\log\log (NL))$ by Lemmas \ref{Lamma:complexity33} and \ref{Lamma:complexity11}.
Remarkably, similar to worker computation complexity, the terms $\prod_{g'\in[G]\backslash\{g\}}\prod_{k\in[R]}(\beta_{g,r}-\beta_{g',k}),g\in[G],r\in[R]$ in \eqref{GLCC:evaluations} are constant and can be pre-computed at the master, and thus the complexity of eliminating these terms is also negligible.

\section{Application to Distributed Model Training and Experiments}\label{experiments}
In this section, we present a practical application of GLCC codes in distributedly training a single-layer perceptron while keeping the training data private. We experimentally demonstrate the performance gain of GLCC in speeding up the training process in the presence of stragglers. 

\subsection{Applying GLCC Codes to Train Perceptron}
A single-layer perceptron is 
a neural network represented by a forward function 
\begin{IEEEeqnarray*}{c}
y=\left\{
\begin{array}{@{}ll}
1,&\text{if}\; \sigma(\mathbf{x}^{\mathrm{T}}\mathbf{w})=(\mathbf{x}^{\mathrm{T}}\mathbf{w})^2> 0.5\\
0,&\text{otherwise}
\end{array}\right.
\end{IEEEeqnarray*} 
mapping an input $\mathbf{x}\in\mathbb{R}^d$ to a binary output $y\in\{0,1\}$, where $\mathbf{w}\in\mathbb{R}^d$ is a vector of model parameters (or weights), and $\sigma(x)$ is an activation function, and is selected as quadratic polynomial function, i.e., $\sigma(x)=x^2$. 

Given a training dataset denoted by a matrix $\mathbf{X}\in\mathbb{R}^{s \times d}$ consisting of $s$ data points with $d$ features, and a corresponding label vector $\mathbf{y}\in\{0,1\}^{s}$, the goal 
is to train the model parameters $\mathbf{w}$ by minimizing the following mean-squared error (MSE) loss function,
\begin{IEEEeqnarray}{c}\notag
C(\mathbf{w}) = \frac{1}{m} \big\Vert \left(\mathbf{X}\mathbf{w}\right)^2-\mathbf{y} \big\Vert_{2}^2,
\end{IEEEeqnarray}
where $(\cdot)^2$ operates element-wise over the vector given by $\mathbf{X}\mathbf{w}$. 


Training is performed by iterating a gradient decent rule that moves the model along the negative gradient direction. The gradient of the loss function $C(\mathbf{w})$ is given by 
\begin{IEEEeqnarray}{c}
\nabla C(\mathbf{w}) =
\frac{4}{m}\left[\mathbf{X}^{\mathrm{T}} (\mathbf{X} \mathbf{w})^3 - \mathbf{X}^{\mathrm{T}} (\mathbf{X} \mathbf{w} \circ \mathbf{y})\right], \label{gradient:func} 
\end{IEEEeqnarray}
where $\circ$ denotes Hadamard multiplication.
Let $\mathbf{w}^{(t)}$ denote the model at iteration $t$ and $\eta$ be the learning rate. Then the model is updated as
\begin{IEEEeqnarray}{c}\label{Gradient:updata22}
\mathbf{w}^{(t+1)}\!=\!\mathbf{w}^{(t)}\!-\!
\frac{4\eta}{m}\left[\mathbf{X}^{\mathrm{T}} (\mathbf{X} \mathbf{w}^{(t)})^3 \!-\! \mathbf{X}^{\mathrm{T}} (\mathbf{X} \mathbf{w}^{(t)} \circ \mathbf{y})\right].\IEEEeqnarraynumspace
\end{IEEEeqnarray}

To demonstrate the performance of the proposed GLCC codes, we consider simultaneously training $M$ binary classifiers over the $M$ datasets $(\mathbf{X}_1,\mathbf{y}_1),\ldots,(\mathbf{X}_M,\mathbf{y}_M)$, using the single-layer perceptron described above. 
Training is performed by offloading computationally-intensive operations to a distributed computing system with $N$ workers.  
In doing this, the master wishes to protect the data $(\mathbf{X}_m,\mathbf{y}_m,\mathbf{w}_m^{(t)})_{m\in[M]}$ private from any $T$ colluding workers,\footnote{Similar to \cite{LCC,so2021codedprivateml,so2020scalable}, we do not consider the case of adversarial workers in our experiments for simplifying the comparison.} where $\mathbf{w}_m^{(t)}$ is the model for the dataset $(\mathbf{X}_m,\mathbf{y}_m)$ at iteration $t$.\footnote{It is necessary to provide a privacy guarantee for the model parameters because references  \cite{zhu2019deep,wang2019beyond,geiping2020inverting} have shown that the reconstruction of the training samples from the model is possible by using a model inversion attack.}

For the training problem, the computationally-intensive operations correspond to the gradient computations in \eqref{gradient:func}. 
GLCC codes are able to deal with the gradient computation step while protecting the privacy of the data, i.e., evaluating a multivariate polynomial $\phi$ at the $M$ points $(\mathbf{X}_m,\mathbf{y}_m,\mathbf{w}_m^{(t)})_{m\in[M]}$ for the $t$-th iteration, where function $\phi$ is given by 
\begin{IEEEeqnarray}{c}\label{exp:polynomial}
    \phi(\mathbf{X,y,w})=\mathbf{X}^{\mathrm{T}} (\mathbf{X} \mathbf{w})^3 - \mathbf{X}^{\mathrm{T}} (\mathbf{X} \mathbf{w} \circ \mathbf{y}).
\end{IEEEeqnarray}
However, one obstacle to applying GLCC codes is that they are designed for computations over a finite field, while the dataset $\mathbf{X}_m$ and model $\mathbf{w}^{(t)}_m$ are distributed over the domain of real number.
Similar to \cite{so2021codedprivateml,so2020scalable,so2020byzantine}, our solution is to quantize the dataset and model from the real number domain to the domain of integers and then embed it in the finite field $\mathbb{F}_q$ with prime $q$.

Define a quantization function as
\begin{IEEEeqnarray}{c}\label{fun:quan}
\textit{Q}(x;l)=\left\{
\begin{array}{@{}ll}
\lfloor2^{l}\cdot x\rceil,& \text{if}\; x\geq 0\\
q+\lfloor 2^{l}\cdot x\rceil,& \text{if}\; x< 0
\end{array}\right.,
\end{IEEEeqnarray}
where  $l$ is a positive integer parameter that controls the quantization precision, and $\lfloor\cdot\rceil$ is a function of rounding half up.
Then the dataset $\mathbf{X}_m$ and the model $\mathbf{w}_m^{(t)}$ are quantized as
\begin{IEEEeqnarray}{c}\label{quantization:E}
\overline{\mathbf{X}}_m=\textit{Q}(\mathbf{X}_m;l_x), \;\;
\overline{\mathbf{w}}_m^{(t)}=\textit{Q}(\mathbf{w}_m^{(t)};l_w), \;\; \forall\, m\in[M],\IEEEeqnarraynumspace
\end{IEEEeqnarray}  
where function $\textit{Q}$ is carried out element-wise, and $l_x$ and $l_w$ are the quantization parameters of  $\mathbf{X}_m$ and  $\mathbf{w}_m^{(t)}$, respectively.
Notably, given the quantized versions of $\mathbf{X}_m$ and $\mathbf{w}_m^{(t)}$ \eqref{quantization:E}, the terms $\overline{\mathbf{X}}^{\mathrm{T}}_m(\overline{\mathbf{X}}_m\overline{\mathbf{w}}^{(t)}_m)^3$ and $\overline{\mathbf{X}}^{\mathrm{T}}_m(\overline{\mathbf{X}}_m\overline{\mathbf{w}}^{(t)}_m\circ{\mathbf{y}}_m)$ of the polynomial function \eqref{exp:polynomial} have distinct quantization precision, which will result in large quantization loss when using some specific dequantization function to dequantize the difference between the two terms (i.e., $\overline{\mathbf{X}}^{\mathrm{T}}_m(\overline{\mathbf{X}}_m\overline{\mathbf{w}}^{(t)}_m)^3- \overline{\mathbf{X}}^{\mathrm{T}}_m(\overline{\mathbf{X}}_m\overline{\mathbf{w}}^{(t)}_m\circ{\mathbf{y}}_m)$). 
To resolve this issue, we further quantize the label  $\mathbf{y}_m$ as   
\begin{IEEEeqnarray}{c}\label{quantization:y}
\overline{\mathbf{y}}_m=\textit{Q}(\mathbf{y}_m;l_y), \quad \forall\, m\in[M]
\end{IEEEeqnarray}
with $l_y=2l_x+2l_w$, such that the two terms $\overline{\mathbf{X}}^{\mathrm{T}}_m(\overline{\mathbf{X}}_m\overline{\mathbf{w}}^{(t)}_m)^3$ and $\overline{\mathbf{X}}^{\mathrm{T}}_m(\overline{\mathbf{X}}_m\overline{\mathbf{w}}^{(t)}_m\circ\overline{\mathbf{y}}_m)$ have same quantization precision of $4l_x+3l_w$.

Now, GLCC codes can be employed to compute the evaluations of function  $\phi$ at the $M$ quantized points $X_1=(\overline{\mathbf{X}}_1,\overline{\mathbf{y}}_1,\overline{\mathbf{w}}_1^{(t)}),\ldots,X_M=(\overline{\mathbf{X}}_M,\overline{\mathbf{y}}_M,\overline{\mathbf{w}}_M^{(t)})$ for the $t$-th iteration. Remarkably, the data $(\overline{\mathbf{X}}_m,\overline{\mathbf{y}}_m)_{m\in[M]}$ are independent of the iteration index $t$, and thus they need to be shared at the distributed computing system only once before the training starts. Since $\phi$ is a  multivariate polynomial of degree $D=7$ in variables $(\mathbf{X,y,w})$ by \eqref{exp:polynomial}, the master can recover the desired gradient computation
\begin{IEEEeqnarray}{c}
\phi(\overline{\mathbf{X}}_m,\overline{\mathbf{y}}_m,\overline{\mathbf{w}}_m^{(t)})\!=\!\overline{\mathbf{X}}^{\mathrm{T}}_m\!\!\left(\overline{\mathbf{X}}_m\overline{\mathbf{w}}^{(t)}_m\!\right)^3\!-\! \overline{\mathbf{X}}^{\mathrm{T}}_m\!\!\left(\overline{\mathbf{X}}_m\overline{\mathbf{w}}^{(t)}_m\!\circ\!\overline{\mathbf{y}}_m\!\right) \IEEEeqnarraynumspace \label{Gra:computation}
\end{IEEEeqnarray}
for all $m\in[M]$, as long as it receives the responses from any $K=\big\lceil\frac{6(M-G)+GM}{GL}+7T\big\rceil$ workers for any $G,L$ with $G|M$ by Theorem \ref{GLCC:theorem}. 

Let $\textit{Q}^{-1}(x)$ be a dequantization function that converts the input $x$ from the finite field to the real domain, given by 
\begin{IEEEeqnarray}{c}\label{fun:dequan}
\textit{Q}^{-1}(x)=\left\{
\begin{array}{@{}ll}
2^{-l}\cdot x,& \text{if}\; 0\leq x<\frac{q-1}{2} \\
2^{-l}\cdot(x-q),& \text{if}\; \frac{q-1}{2}\leq x<q
\end{array}\right.,
\end{IEEEeqnarray}
where $l=4l_x+3l_w$, and the input $x$ is an element on the finite field $\mathbb{F}_q$. 

Let $\widehat{\mathbf{X}}_m$ and $\widehat{\mathbf{w}}_m^{(t)}$ represent the dequantization versions of the quantized dataset $\overline{\mathbf{X}}_m$ and  model $\overline{\mathbf{w}}_m^{(t)}$, respectively, i.e.,
\begin{IEEEeqnarray}{c}\label{qun:loss}
    \widehat{\mathbf{X}}_m=2^{-l_x}\big\lfloor2^{l_x}\mathbf{X}_m\big\rceil,\quad\quad
    \widehat{\mathbf{w}}_m^{(t)}=2^{-l_w}\big\lfloor2^{l_w}\mathbf{w}_m^{(t)}\big\rceil. \IEEEeqnarraynumspace
\end{IEEEeqnarray}
The parameters $l_x$ and $l_w$ control the quantization losses $|{\mathbf{X}}_m\!-\!\widehat{\mathbf{X}}_m|$ and $|{\mathbf{w}}_m^{(t)}\!-\!\widehat{\mathbf{w}}_m^{(t)}|$ of the dataset and model, respectively. Larger values of $l_x$ and $l_w$ reduce the quantization losses.

Then the computational gradient  $\phi(\overline{\mathbf{X}}_m,\overline{\mathbf{y}}_m,\overline{\mathbf{w}}_m^{(t)})$  \eqref{Gra:computation} is dequantized from the finite field to the real domain, given by
\begin{IEEEeqnarray}{l}
\textit{Q}^{-1}\left(\overline{\mathbf{X}}^{\mathrm{T}}_m\left(\overline{\mathbf{X}}_m\overline{\mathbf{w}}^{(t)}_m\right)^3- \overline{\mathbf{X}}^{\mathrm{T}}_m\left(\overline{\mathbf{X}}_m\overline{\mathbf{w}}^{(t)}_m\circ\overline{\mathbf{y}}_m\right)\right)\notag\\
\quad\quad\quad\quad=\widehat{\mathbf{X}}^{\mathrm{T}}_m\left(\widehat{\mathbf{X}}_m\widehat{\mathbf{w}}^{(t)}_m\right)^3- \widehat{\mathbf{X}}^{\mathrm{T}}_m\left(\widehat{\mathbf{X}}_m\widehat{\mathbf{w}}^{(t)}_m\circ{\mathbf{y}}_m\right), \label{dequan:process}\IEEEeqnarraynumspace
\end{IEEEeqnarray}
which is proved in the following lemma.

\begin{Lemma}\label{size:field}
For the task of training a single-layer perceptron, GLCC codes can complete the desired gradient computations 
$\widehat{\mathbf{X}}^{\mathrm{T}}_m\big(\widehat{\mathbf{X}}_m\widehat{\mathbf{w}}^{(t)}_m\big)^3- \widehat{\mathbf{X}}^{\mathrm{T}}_m\big(\widehat{\mathbf{X}}_m\widehat{\mathbf{w}}^{(t)}_m\circ{\mathbf{y}}_m\big), m\in[M]$ with a certain level of quantization loss as long as the size $q$ of the finite field operated by GLCC codes is large enough such that
\begin{IEEEeqnarray}{l}
  \max\!\bigg(\bigg|  \big\lfloor2^{l_x}\mathbf{X}_m^{\mathrm{T}}\big\rceil\!\Big(\big\lfloor2^{l_x}\mathbf{X}_m^{\mathrm{T}}\big\rceil\!\cdot\!\big\lfloor2^{l_w}\mathbf{w}_m^{(t)}\big\rceil\Big)^{3}\!-\notag\\
\;\;\!\big\lfloor2^{l_x}\mathbf{X}_m^{\mathrm{T}}\big\rceil\!\Big(\big\lfloor2^{l_x}\mathbf{X}_m^{\mathrm{T}}\big\rceil\!\cdot\!\big\lfloor2^{l_w}\mathbf{w}_m^{(t)}\big\rceil\Big)\!\circ\!\big\lfloor2^{l_y}\mathbf{y}_m\big\rceil\bigg|\bigg)
  \!<\!\frac{q-1}{2},\IEEEeqnarraynumspace \label{round:error}
\end{IEEEeqnarray}
where all the computational operations are performed on integer domains; $|\cdot|$ and $\max(\cdot)$ denote taking the absolute value element-wise and the maximum entity of the given matrix, respectively. 
\end{Lemma}
\begin{proof}
For the dequantization function $\textit{Q}^{-1}(x)$ defined in \eqref{fun:dequan}, the input $x$ can be also negative number due to $x= q+x\;(\text{mod $q$})$ if $x<0$. Thus $\textit{Q}^{-1}(x)$ can be equivalently represented as 
\begin{IEEEeqnarray}{c}\label{sim:deq}
   \textit{Q}^{-1}(x)=2^{-l}\cdot x, \quad\quad\text{for}\;-\frac{q+1}{2}\leq x<\frac{q-1}{2}. 
\end{IEEEeqnarray}

Then for all $m\in[M]$, we have
\begin{IEEEeqnarray}{rCl}
&&\textit{Q}^{-1}\left(\overline{\mathbf{X}}^{\mathrm{T}}_m\left(\overline{\mathbf{X}}_m\overline{\mathbf{w}}^{(t)}_m\right)^3- \overline{\mathbf{X}}^{\mathrm{T}}_m\left(\overline{\mathbf{X}}_m\overline{\mathbf{w}}^{(t)}_m\circ\overline{\mathbf{y}}_m\right)\right) \notag\\
&\overset{(a)}{=}&\textit{Q}^{-1}\Big(\big\lfloor2^{l_x}\mathbf{X}_m^{\mathrm{T}}\big\rceil\left(\big\lfloor2^{l_x}\mathbf{X}_m\big\rceil\cdot\big\lfloor2^{l_w}\mathbf{w}_m^{(t)}\big\rceil\right)^{3}\notag\\
&&\quad\quad-\big\lfloor2^{l_x}\mathbf{X}_m^{\mathrm{T}}\big\rceil\left(\big\lfloor2^{l_x}\mathbf{X}_m\big\rceil\cdot\big\lfloor2^{l_w}\mathbf{w}_m^{(t)}\big\rceil\circ\big\lfloor2^{l_y}\mathbf{y}_m\big\rceil\right)\Big) \notag\\
&\overset{(b)}{=}&2^{-(4l_x+3l_w)}\left(\big\lfloor2^{l_x}\mathbf{X}_m^{\mathrm{T}}\big\rceil\left(\big\lfloor2^{l_x}\mathbf{X}_m\big\rceil\cdot\big\lfloor2^{l_w}\mathbf{w}_m^{(t)}\big\rceil\right)^{3} \right)\notag \\
&&\quad\quad\quad-2^{-(4l_x+3l_w)}\Big(\big\lfloor2^{l_x}\mathbf{X}_m^{\mathrm{T}}\big\rceil\bigg(\big\lfloor2^{l_x}\mathbf{X}_m\big\rceil\cdot\big\lfloor2^{l_w}\mathbf{w}_m^{(t)}\big\rceil\notag\\
&&\quad\quad\quad\quad\quad\quad\quad\quad\quad\quad\quad\quad\quad\circ\big\lfloor2^{2l_x+2l_w}\mathbf{y}_m\big\rceil\bigg)\Big) \notag \\
&=&2^{-l_x}\big\lfloor2^{l_x}\mathbf{X}_m^{\mathrm{T}}\big\rceil\left(2^{-l_x}\big\lfloor2^{l_x}\mathbf{X}_m\big\rceil\cdot2^{-l_w}\big\lfloor2^{l_w}\mathbf{w}_m^{(t)}\big\rceil\right)^{3}\notag\\
&&\quad\quad\quad-2^{-l_x}\big\lfloor2^{l_x}\mathbf{X}_m^{\mathrm{T}}\big\rceil\Big(2^{-l_x}\big\lfloor2^{l_x}\mathbf{X}_m\big\rceil\cdot2^{-l_w}\big\lfloor2^{l_w}\mathbf{w}_m^{(t)}\big\rceil\notag\\
&&\quad\quad\quad\quad\quad\quad\quad\quad\quad\quad\quad\circ2^{-(2l_x+2l_w)}\big\lfloor2^{2l_x+2l_w}\mathbf{y}_m\big\rceil\Big) \notag\\
&\overset{(c)}{=}&\widehat{\mathbf{X}}^{\mathrm{T}}_m\left(\widehat{\mathbf{X}}_m\widehat{\mathbf{w}}^{(t)}_m\right)^3- \widehat{\mathbf{X}}^{\mathrm{T}}_m\left(\widehat{\mathbf{X}}_m\widehat{\mathbf{w}}^{(t)}_m\circ{\mathbf{y}}_m\right),\notag
\end{IEEEeqnarray}
where $(a)$ is due to $\overline{\mathbf{X}}^{\mathrm{T}}_m\big(\overline{\mathbf{X}}_m\overline{\mathbf{w}}^{(t)}_m\big)^3\!-\! \overline{\mathbf{X}}^{\mathrm{T}}_m\big(\overline{\mathbf{X}}_m\overline{\mathbf{w}}^{(t)}_m\circ\overline{\mathbf{y}}_m\big)
\!=\!\big\lfloor2^{l_x}\mathbf{X}_m^{\mathrm{T}}\big\rceil\big(\big\lfloor2^{l_x}\mathbf{X}_m\big\rceil\cdot\big\lfloor2^{l_w}\mathbf{w}_m^{(t)}\big\rceil\big)^{3}-\big\lfloor2^{l_x}\mathbf{X}_m^{\mathrm{T}}\big\rceil\big(\big\lfloor2^{l_x}\mathbf{X}_m\big\rceil\cdot\big\lfloor2^{l_w}\mathbf{w}_m^{(t)}\big\rceil\circ\big\lfloor2^{l_y}\mathbf{y}_m\big\rceil\big)\;(\text{mod $q$})$ on the finite field $\mathbb{F}_q$ by \eqref{fun:quan}, \eqref{quantization:E} and \eqref{quantization:y}; $(b)$ follows by \eqref{round:error} and \eqref{sim:deq}; and $(c)$ is due to \eqref{qun:loss} and the fact that the label $\mathbf{y}_m$ has no quantization loss because its elements are over the domain of integers such that 
$2^{-(2l_x+2l_w)}\big\lfloor2^{2l_x+2l_w}\mathbf{y}_m\big\rceil=2^{-(2l_x+2l_w)}\cdot 2^{2l_x+2l_w}\mathbf{y}_m=\mathbf{y}_m$.  This completes the proof of the lemma.
\end{proof}

Finally, the current model $\overline{\mathbf{w}}^{(t)}_m,m\in[M]$ is updated via \eqref{Gradient:updata22} for the next iteration. 
\begin{IEEEeqnarray}{l}
\mathbf{w}^{(t+1)}_m\!=\!\mathbf{w}^{(t)}_m\!-\!
\frac{4\eta}{m}\!\bigg(\!\widehat{\mathbf{X}}^{\mathrm{T}}_m\!\Big(\!\widehat{\mathbf{X}}_m\widehat{\mathbf{w}}^{(t)}_m\!\Big)^3\!-\! \widehat{\mathbf{X}}^{\mathrm{T}}_m\!\Big(\!\widehat{\mathbf{X}}_m\widehat{\mathbf{w}}^{(t)}_m\circ{\mathbf{y}}_m\!\Big)\!\bigg).\notag\\
\label{Gradient:updata}
\end{IEEEeqnarray}
Obviously, quantization losses are the only concerns for the distributed gradient computations using GLCC codes, as long as \eqref{round:error} is satisfied. While larger values of the parameters $l_x$ and $l_w$ can effectively reduce the quantization losses, they also raise the risk of an overflow error caused by wrap-around in the finite field $\mathbb{F}_q$ (i.e., violation of the constraint in \eqref{round:error}). Therefore, the parameters $l_x$ and $l_w$ involve a trade-off between the quantization losses and the risk of overflow error, serving as  hyperparameters that need to be carefully selected.

To summarize, the overall procedures described above are outlined in Algorithm \ref{GLCC:application}.

\begin{algorithm}[t]
\caption{Applying GLCC Codes to Train Perceptron}
\label{GLCC:application}
\begin{algorithmic}[1] 
\REQUIRE Dataset $(\mathbf{X}_1,\mathbf{y}_1),\ldots,(\mathbf{X}_M,\mathbf{y}_M)$, design parameters $G,L$ with $G|M$, quantization parameters $l_x,l_w$, learning rate $\eta$, and number of iterations $J$.
\ENSURE Model parameters $\mathbf{w}^{(J)}_1,\ldots,\mathbf{w}^{(J)}_M$.
\STATE Compute the quantized versions $\overline{\mathbf{X}}_m$ and $\overline{\mathbf{y}}_m$ of data $\mathbf{X}_m$ and label $\mathbf{y}_m$ by \eqref{quantization:E} and \eqref{quantization:y}, respectively, for all $m\in[M]$.
\STATE Encode the data $X_1=(\overline{\mathbf{X}}_1,\overline{\mathbf{y}}_1),\ldots,X_M=(\overline{\mathbf{X}}_M,\overline{\mathbf{y}}_M)$ into  $\widetilde{X}_{n}^D$ according to \eqref{encdoing function} and \eqref{GLCC:upload}, and then send $\widetilde{X}_{n}^D$ to worker $n$ for all $n\in[N]$.
\STATE Initialize the models $\mathbf{w}^{(0)}_1,\ldots,\mathbf{w}^{(0)}_M$ randomly.
\FOR {iteration $t\in[0:J-1]$}
\STATE Master quantizes model $\mathbf{w}^{(t)}_m$ into $\overline{\mathbf{w}}^{(t)}_m$ for all $m\in[M]$ by \eqref{quantization:E}.
\STATE Master encodes the models $\overline{\mathbf{w}}^{(t)}_1,\ldots,\overline{\mathbf{w}}^{(t)}_M$ into  $\widetilde{w}_{n}^{(t)}$ according to \eqref{encdoing function} and \eqref{GLCC:upload}, and
then send $\widetilde{w}_{n}^{(t)}$ to worker $n$ for all $n\in[N]$.
\FOR {$n=1,2,\ldots,N$}\label{for:split}
\STATE Worker $n$ computes $\widetilde{Y}_n=(\widetilde{Y}_{n,1},\ldots,\widetilde{Y}_{n,L})$ over the received encoding data $\widetilde{X}_{n}=(\widetilde{X}_{n}^{D},\widetilde{w}_{n}^{(t)})$ according to \eqref{GLCC:download} and \eqref{exp:polynomial} and then send back to the master, where $\widetilde{X}_{n}=(\widetilde{X}_{n}^{D},\widetilde{w}_{n}^{(t)})$ is equivalent to the encoded version of the data $X_1=(\overline{\mathbf{X}}_1,\overline{\mathbf{y}}_1,\overline{\mathbf{w}}^{(t)}_1),\ldots,X_M=(\overline{\mathbf{X}}_M,\overline{\mathbf{y}}_M,\overline{\mathbf{w}}^{(t)}_M)$ according to \eqref{encdoing function} and \eqref{GLCC:upload}.
\ENDFOR
\IF{the master receives any fastest $K=\big\lceil\frac{6(M-G)+GM}{GL}+7T\big\rceil$ responses}
\STATE Decode the polynomial $h(x)$ \eqref{GLCC:answer fun} from the received responses by using RS decoding algorithm, and recover the gradient computations $\{\phi(\overline{\mathbf{X}}_m,\overline{\mathbf{y}}_m,\overline{\mathbf{w}}^{(t)}_m)\}_{m\in[M]}$ from \eqref{recover:result}.
\ENDIF
\STATE Master converts $\phi(\overline{\mathbf{X}}_m,\overline{\mathbf{y}}_m,\overline{\mathbf{w}}^{(t)}_m)$ from finite field to real domain and updates the model to obtain $\mathbf{w}^{(t+1)}_m$ according to \eqref{dequan:process} and \eqref{Gradient:updata} for all $m\in[M]$.
\ENDFOR
\RETURN Models $\mathbf{w}^{(J)}_1,\ldots,\mathbf{w}^{(J)}_M$.
\end{algorithmic}
\end{algorithm}

\subsection{Experiments}\label{expe:results}
In this subsection, we experimentally evaluate the performance of GLCC codes on image classification tasks, and compare them with those of the baseline LCC codes.\footnote{Notably in reference \cite{so2021codedprivateml}, a remarkable approach to distributed training of a logistic regression model is proposed using the idea of LCC codes. 
For a fair comparison, we just apply LCC codes to the single-layer perceptron network described in the previous subsection as a baseline.}


\noindent {\bf Datasets and models.} We focus on using the single-layer perceptron to simultaneously train $M=5$ binary classifiers on two benchmark image datasets: MNIST \cite{lecun1998mnist} and CIFAR-10 \cite{krizhevsky2009learning}. 
For the MNIST dataset, the $5$ binary classifiers are used to distinguish digits 0 and 1, 2 and 3, 4 and 5, 6 and 7, and 8 and 9, respectively. The training set for each classifier consists of $s=11200$ grayscale images of size $28\times28$.
For the CIFAR-10 dataset, we adopt an ImageNet pre-trained VGG19 model \cite{simonyan2014very} to extract a $512$-dimensional feature representation of each color image.
The $5$ binary classifiers on the CIFAR-10 dataset are trained to distinguish images with the labels of airplane and cat, frog and automobile, deer and ship, bird and horse, and dog and truck, respectively,
where the training set of each classifier contains $s=10000$ feature representations. 

\noindent {\bf Setup.}
Our experiments are performed on a single machine equipped with Intel Xeon Gold 5118 CPU @2.30GHz, simulating a distributed computing system of a master node and $N=50$ worker nodes. 
We also simulate the master-worker communication, assuming a total communication bandwidth of $200$Mbps between the master and all the workers. 
We provide a privacy guarantee of the dataset and models from any $T=1$ worker. All experiments are implemented in Python/PyTorch. 
Note that, in our experiments, the encoding and decoding operations at the master node can be achieved through polynomial interpolation and evaluation. Thus the encoding and decoding operations, as well as the gradient computations \eqref{Gra:computation} at the worker nodes, can all be regarded as matrix-matrix multiplications. They are all implemented using \textsl{torch.matmul()} command that is a PyTorch function used to perform matrix multiplication.

To emulate the straggler effect, we inject artificial delays at each training iteration using \textsl{time.sleep()} command. We consider two different scenarios of stragglers. 

\begin{itemize}
    \item \textbf{Straggler Scenario 1:} In the first scenario, we randomly delay the computation
time at each worker for a fixed duration \cite{LCC,GradientCodes1}. Specifically, each worker is chosen to be a straggler with a probability of $0.4$. The worker will wait for $0.05$ seconds before it sends its computation results to the master once it is selected as a straggler.
    \item \textbf{Straggler Scenario 2:} In the second scenario, we consider an exponentially distributed delay \cite{ozfatura2020straggler,wang2019erasurehead}. The length (in seconds) of the delay for each worker is drawn independently from an exponential distribution with parameter $\lambda=2$, i.e., the average delay at each worker is $1/\lambda=0.5$ seconds.
\end{itemize}

\begin{table*}[htbp!]
\centering
  \caption{Running time (seconds) of GLCC and LCC codes on MNIST for $11200$ iterations to achieve target average accuracy $acc = 97.38\%$; and on CIFAR-10 for $20000$ iterations to achieve target average accuracy $acc=90.18\%$.} \label{tab:time}
  \begin{tabular}{c|c|c|c|c|c|c}
\Xcline{1-7}{2pt}
 \multirow{3}{*}{Dataset} & \multirow{3}{*}{\shortstack{Straggler \\ Delay}} & \multirow{3}{*}{Coded Protocol} &  \multirow{3}{*}{\shortstack{Encoding and \\ Decoding}}  & \multirow{3}{*}{\shortstack{Upload and\\ Download}}  & \multirow{3}{*}{\shortstack{Computation \\ at worker}}   & \multirow{3}{*}{Total}  \\ 
  &  & &  &  &   &   \\ 
 &  & &  &  &   &   \\ \Xcline{1-7}{0.8pt}
\multirow{6}{*}{MNIST} & \multirow{3}{*}{\shortstack{Scenario 1}} & LCC codes & 29.23 & 68.69 & 545.27 & 643.19 \\ \Xcline{3-7}{0.4pt}
 & & GLCC with ($G=1,L=2$) & 43.42 & 170.33 & 37.46 & {\bf 251.21} \\ \Xcline{3-7}{0.4pt}
  & & GLCC with ($G=5,L=1$) & 44.17 & 310.58 & 85.72 & 440.47 \\ \Xcline{2-7}{1pt}
 & \multirow{3}{*}{\shortstack{Scenario 2}} & LCC codes & 29.23 & 60.46 & 6243.75 & 6333.44 \\ \Xcline{3-7}{0.4pt}
 & & GLCC with ($G=1,L=2$) & 43.42 & 120.91 & 3535.10 & 3699.43 \\ \Xcline{3-7}{0.4pt}
  & & GLCC with ($G=5,L=1$) & 44.17 & 297.54 & 1372.87 & {\bf 1714.58} \\ \Xcline{1-7}{0.8pt}
\multirow{6}{*}{CIFAR-10} & \multirow{3}{*}{\shortstack{Scenario 1}} & LCC codes & 40.20 & 89.11 & 964.93 & 1094.24 \\ \Xcline{3-7}{0.4pt}
 & & GLCC with ($G=1,L=2$) & 57.38 & 220.69 & 43.27 & {\bf 321.34} \\ \Xcline{3-7}{0.4pt}
  & & GLCC with ($G=5,L=1$) & 60.68 & 402.43 & 102.84 & 565.95 \\ \Xcline{2-7}{1pt}
 & \multirow{3}{*}{\shortstack{Scenario 2}} & LCC codes & 40.20 & 78.34 & 11143.94 & 11262.48 \\ \Xcline{3-7}{0.4pt}
 & & GLCC with ($G=1,L=2$) & 57.38 & 156.67 & 6288.80 & 6502.85 \\ \Xcline{3-7}{0.4pt}
  & & GLCC with ($G=5,L=1$) & 60.68 & 385.54 & 2401.32 & {\bf 2847.54} \\ \Xcline{1-7}{2pt}
  \end{tabular}
\end{table*}

\noindent {\bf Hyperparameters.} 
To avoid overflow error on a $64$-bit implementation, we carefully select the field size $q=2^{27}-39$ on MNIST and $q=2^{30}-35$ on CIFAR-10, and the quantization parameters to $(l_x,l_w)=(0,6)$ and $(0,7)$ for the MNIST and CIFAR-10 dataset respectively. 
We adopt mini-batch momentum SGD to update the models for each iteration, where the size of the local mini-batch samples is set to $100$ on  MNIST and $50$ on CIFAR-10 for each model.

\noindent\textbf{Performance Evaluations.} Recall that GLCC codes include LCC codes as a special case by setting $G=1$ and $L=1$. The values of $G$ and $L$ can be adjusted to optimize the total training time. In our experiments, we consider two parameter cases of ($G=1,L=2$) and ($G=5,L=1$). 
Table \ref{tab:time} presents the breakdowns of the total running times for GLCC codes and LCC codes, for the models to reach certain target average accuracies. 
Compared with LCC codes, GLCC codes tradeoff encoding/decoding and communication time to significantly reduce the worker computation time (dominated by the straggler effect) and the total running time, 
for all parameter cases and straggler scenarios. 
Specifically, in straggler scenario 1, GLCC codes with ($G=1,L=2$) achieve up to $2.5\times$ and $3.4\times$ speedup in the training time over LCC codes for MNIST and CIFAR-10, respectively.
In straggler scenario 2, GLCC codes with ($G=5,L=1$) achieve up to $3.6\times$ and $3.9\times$ speedup in the training time over LCC codes for MNIST and CIFAR-10, respectively.
The performance gain is because GLCC codes achieve lower recovery threshold and thus provide stronger robustness against straggler effect, by allowing the workers to compute and communicate multiple encoded data. 

\begin{figure}[htbp!]
\centering
	\subfigure[MNIST with straggler scenario 1]{\begin{minipage}{0.49\linewidth}
		\centering
		\includegraphics[width=1.0\linewidth]{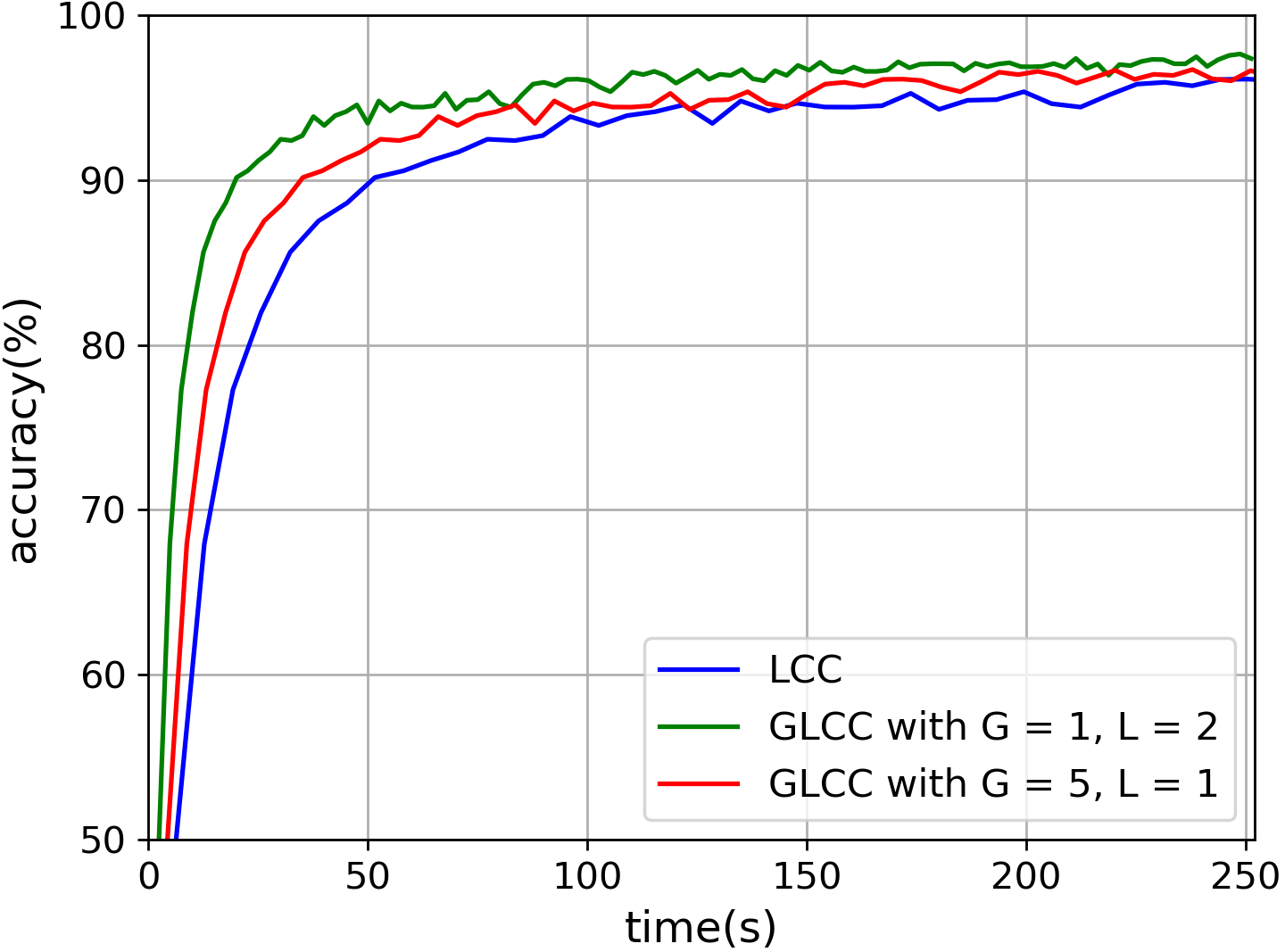}
		\label{Figure_a}
	\end{minipage}}
	\subfigure[CIFAR-10\! with\! straggler\! scenario\! 1]{\begin{minipage}{0.49\linewidth}
		\centering
		\includegraphics[width=1.0\linewidth]{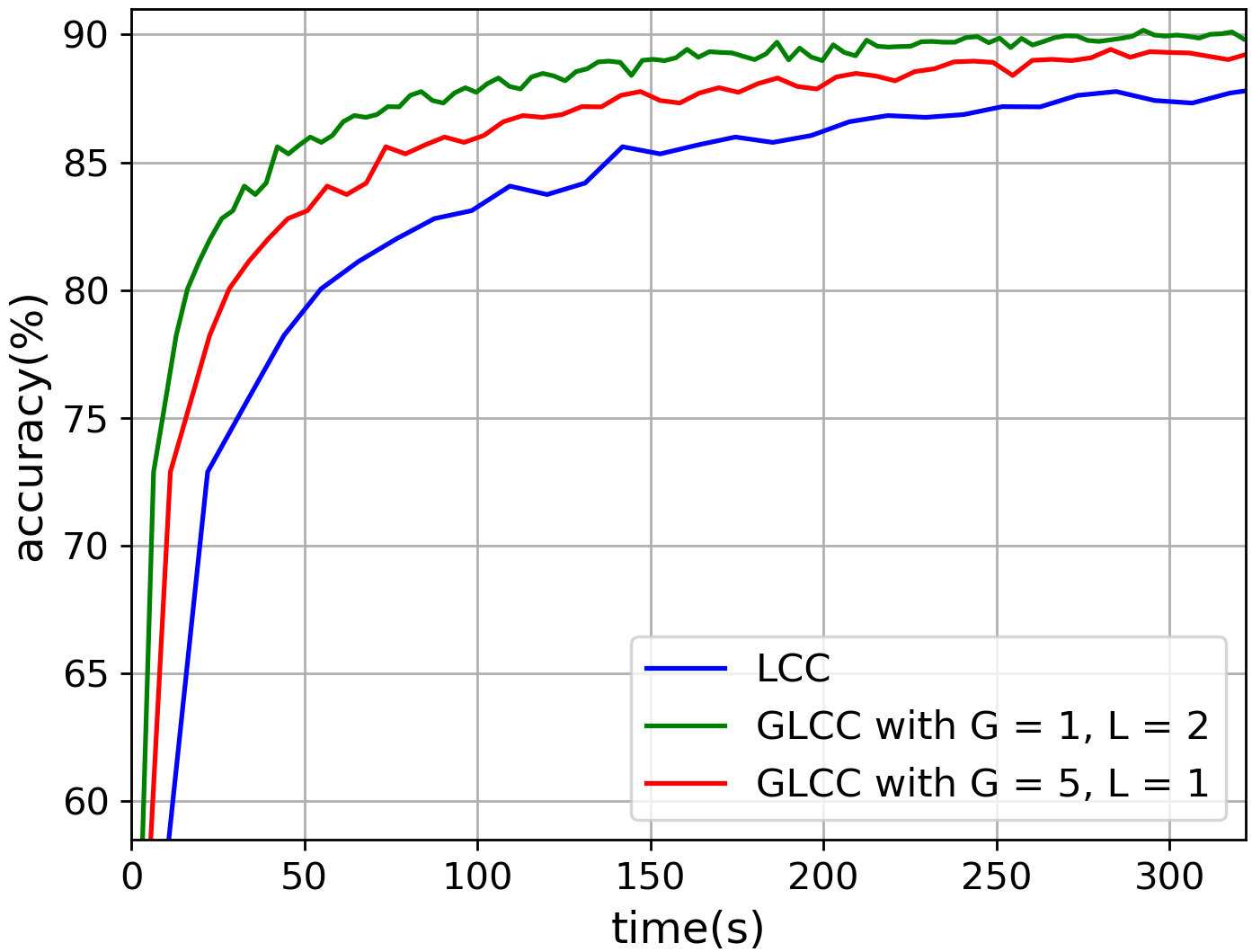}
		\label{Figure_b}
	\end{minipage}}
 
	\subfigure[MNIST with straggler scenario 2]{\begin{minipage}{0.49\linewidth}
		\centering
		\includegraphics[width=1.00\linewidth]{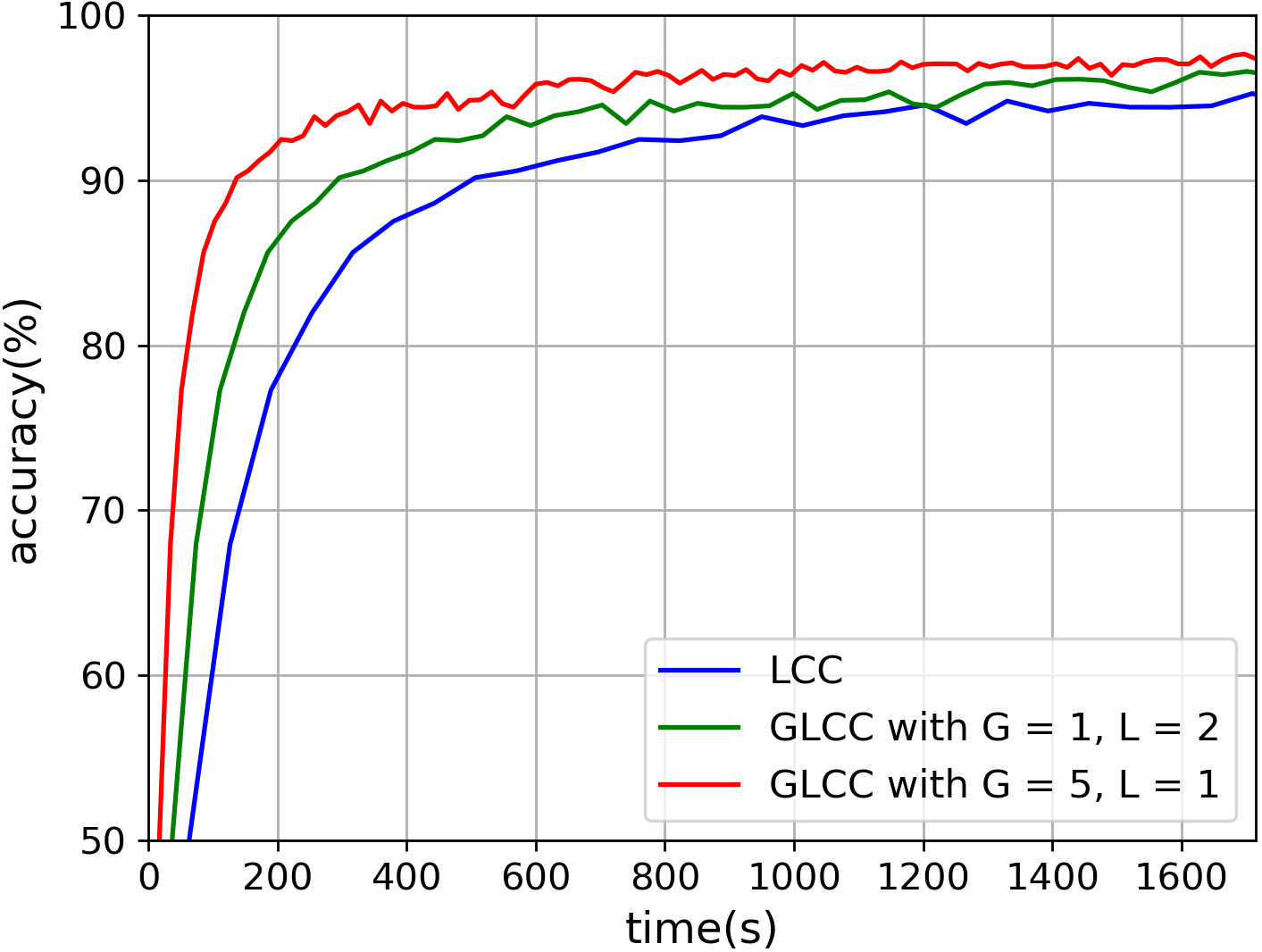}
		\label{Figure_c}
	\end{minipage}}
	\subfigure[CIFAR-10\! with\! straggler\! scenario\! 2]{\begin{minipage}{0.49\linewidth}
		\centering
		\includegraphics[width=1.0\linewidth]{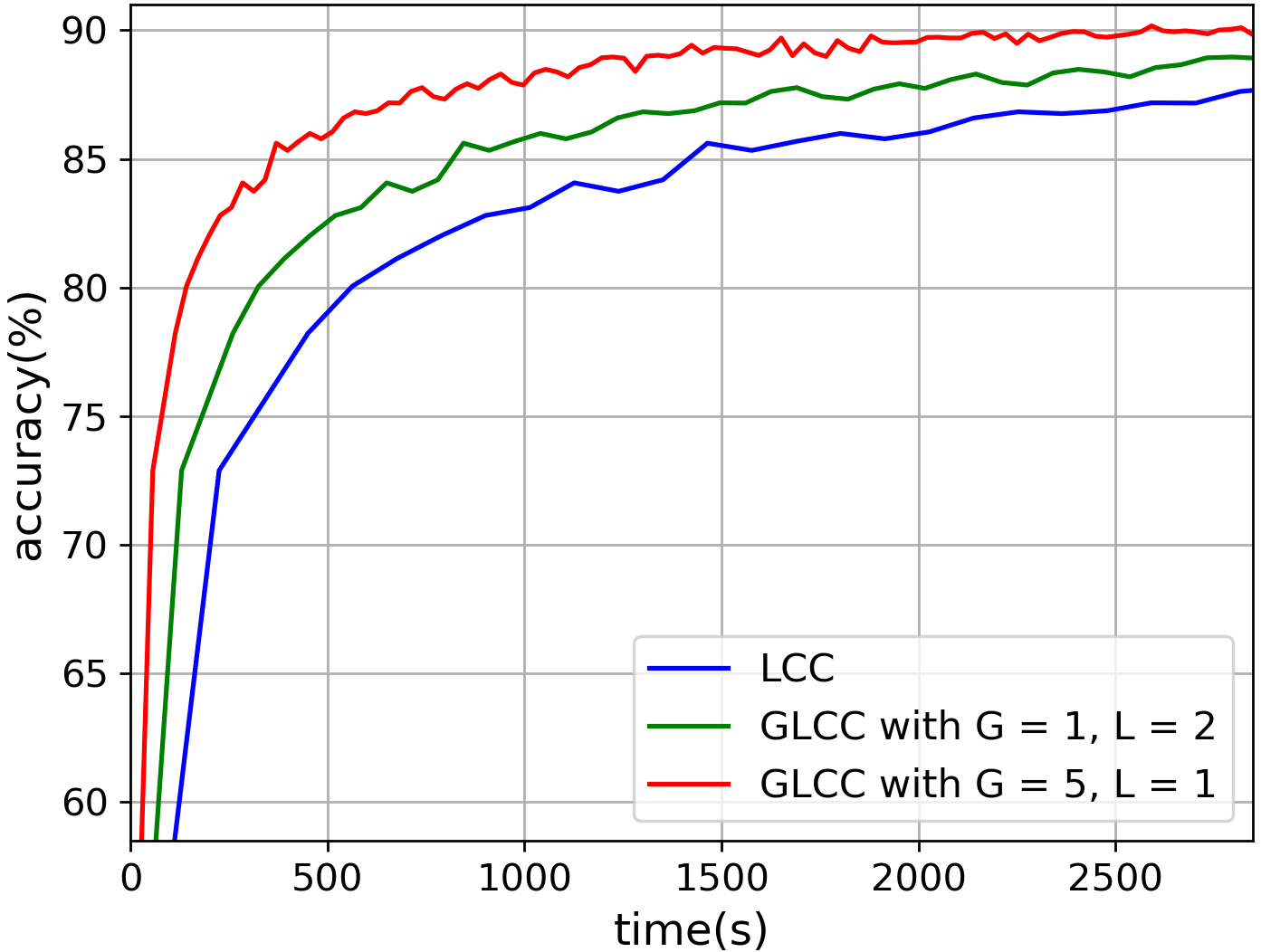}
		\label{Figure_d}
	\end{minipage}}
    \caption{Average test accuracy of GLCC codes and LCC codes to train $M=5$ binary classifiers on MNIST and CIFAR-10 over training time.
    The cut-off time of horizontal coordinate is selected as the convergence time of GLCC codes with ($G=1,L=2$) for subfigures (a)-(b) and GLCC codes with ($G=5,L=1$) for subfigures (c)-(d).
    }
    \label{Figure_convergence}
\end{figure}

Moreover, Fig. \ref{Figure_convergence} plots the average accuracy of GLCC codes and LCC codes along with running time. We observe that GLCC codes outperform LCC codes for both datasets and both straggler scenarios, at any time during the training process. In different straggler scenarios, the configuration of GLCC codes shall be optimized to maximize the training speed. Comparing the $(G=1,L=2)$ and $(G=5,L=1)$ GLCC codes considered in our experiment, the $(G=1,L=2)$ GLCC code has a smaller computation load of $1\times 2=2$ at each worker, but a larger recovery threshold of $22$. In straggler scenario 1 with a smaller straggler delay and a relatively mild straggler possibility, the $(G=1,L=2)$ GLCC is preferable due to its smaller worker computation load; for straggler scenario 2 with a larger straggler delay, the straggler effect becomes more severe and the $(G=5,L=1)$ GLCC code with a smaller recovery threshold of $12$ becomes the referred choice.

\begin{figure}[htbp!]
\centering
	\subfigure[MNIST with straggler scenario 1]{\begin{minipage}{0.49\linewidth}
		\centering
		\includegraphics[width=1.01\linewidth]{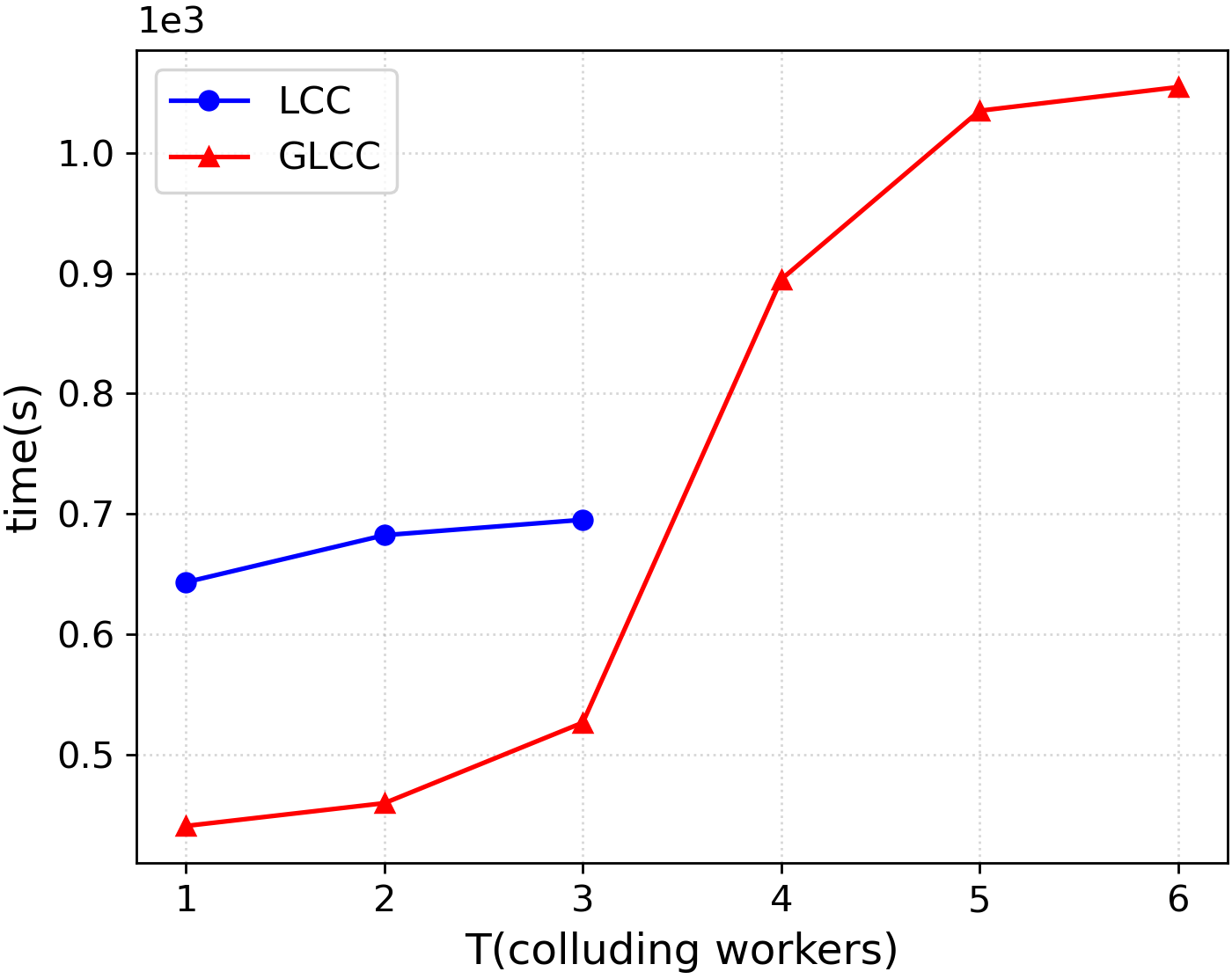}
		\label{Figure:privacy:a}
	\end{minipage}}
	\subfigure[CIFAR-10\! with\! straggler\! scenario\! 1]{\begin{minipage}{0.49\linewidth}
		\centering
		\includegraphics[width=1.01\linewidth]{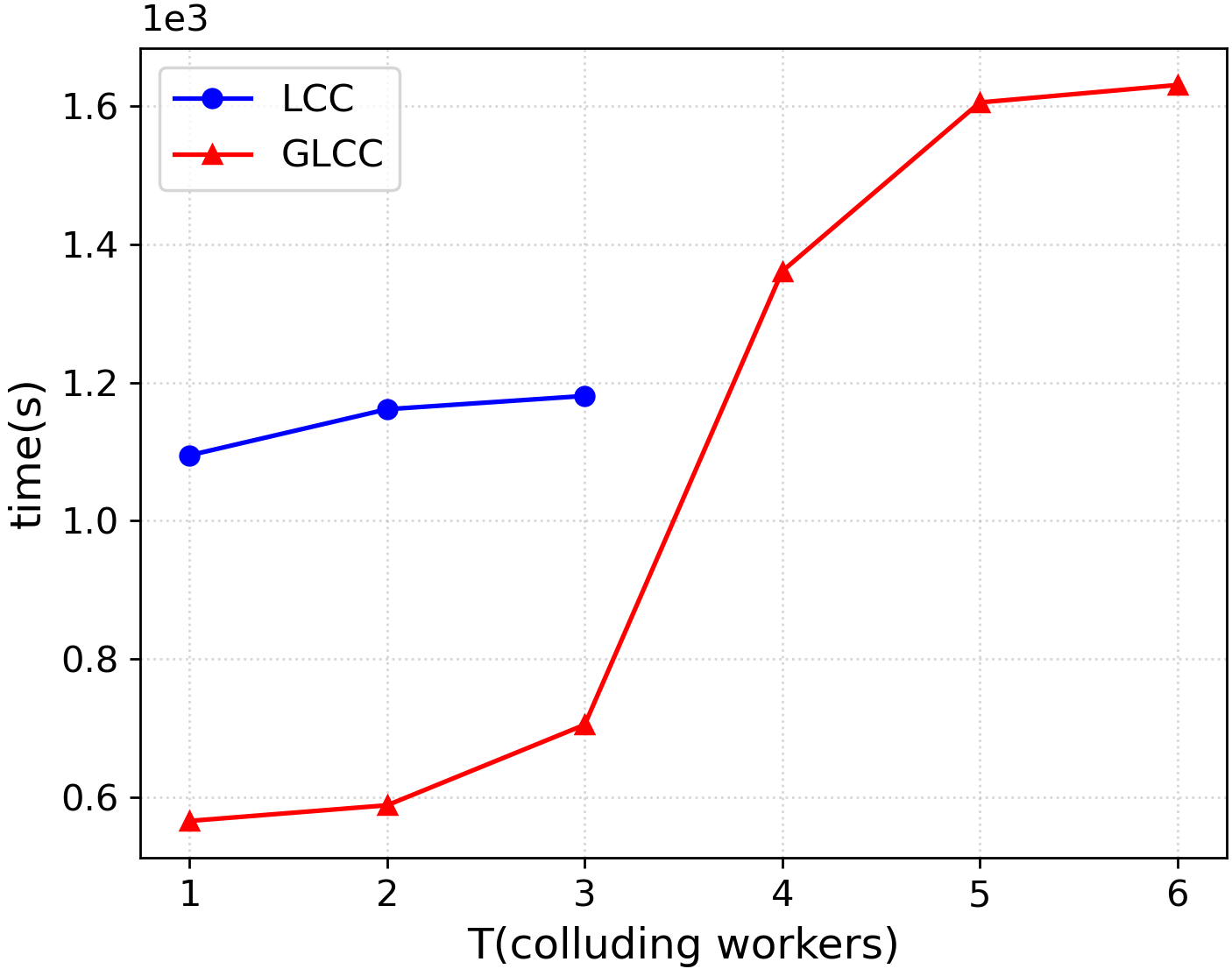}
		\label{Figure:privacy:b}
	\end{minipage}}
 
	\subfigure[MNIST with straggler scenario 2]{\begin{minipage}{0.49\linewidth}
		\centering
		\includegraphics[width=1.01\linewidth]{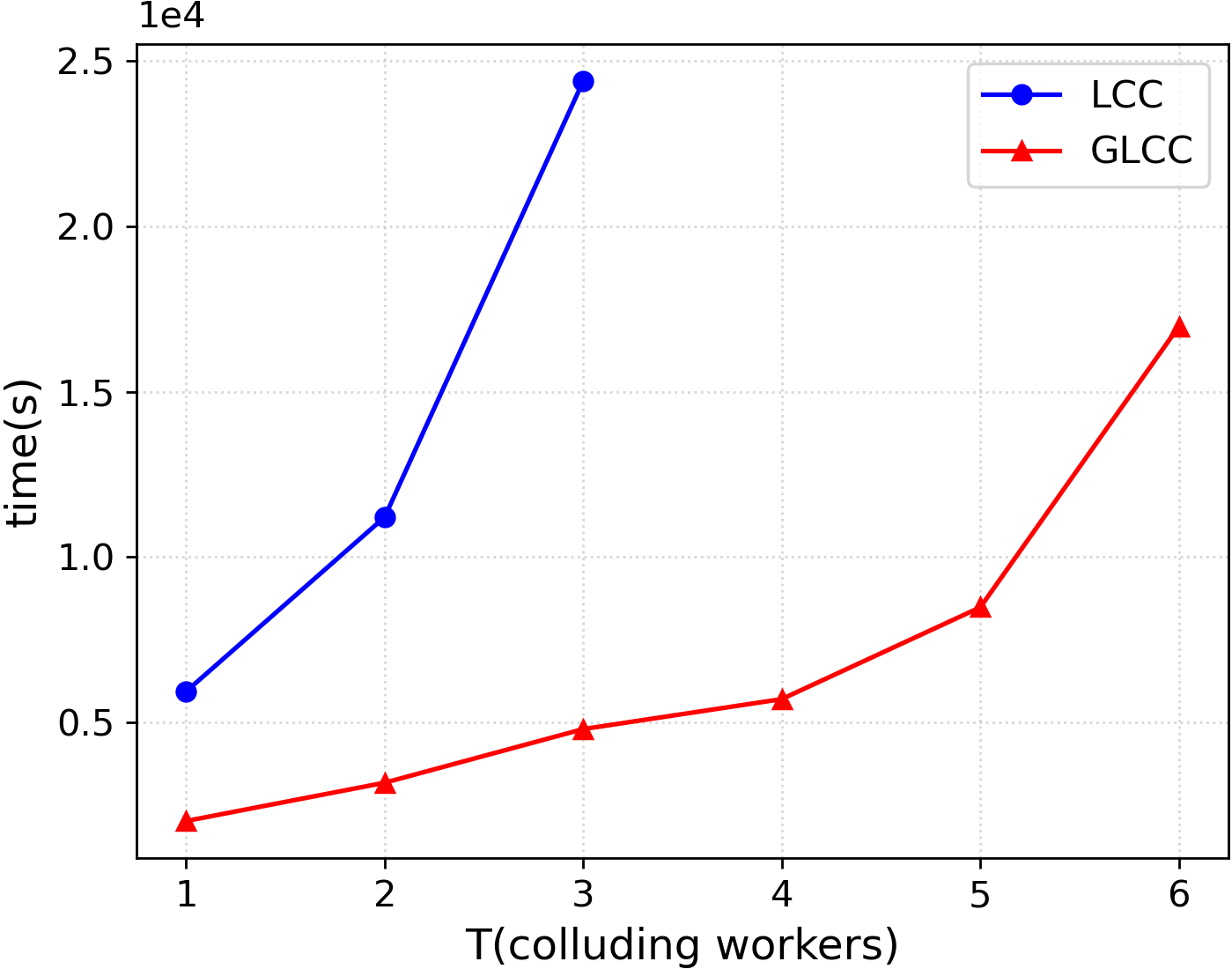}
		\label{Figure:privacy:c}
	\end{minipage}}
	\subfigure[CIFAR-10\! with\! straggler\! scenario\! 2]{\begin{minipage}{0.49\linewidth}
		\centering
		\includegraphics[width=1.01\linewidth]{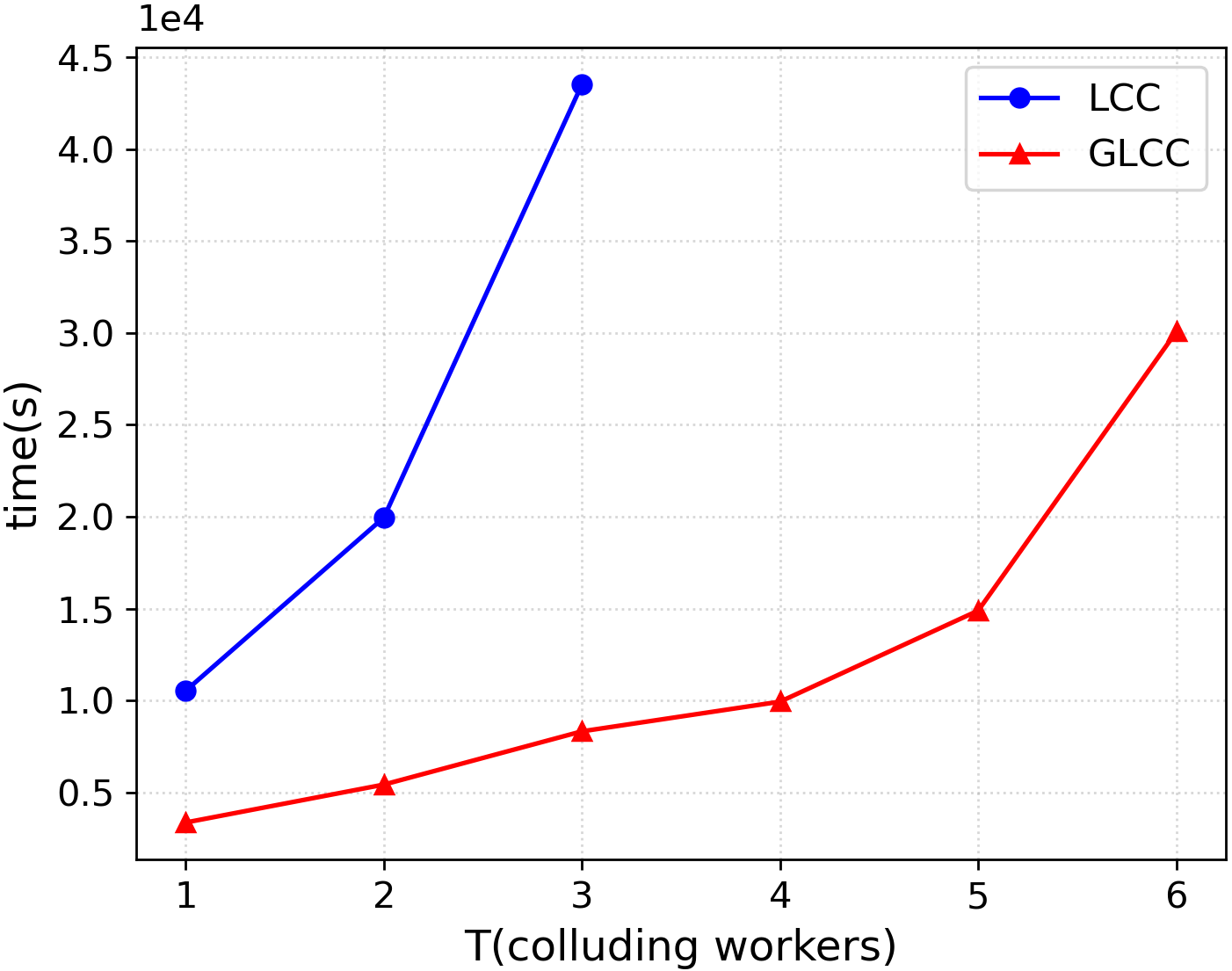}
		\label{Figure:privacy:d}
	\end{minipage}}
    \caption{Running time of GLCC codes with ($G=5,L=1$) and LCC codes with respect to the number of colluding workers $T$. Here $T$ takes the values from $1$ to the maximum value before the privacy guarantee of LCC/GLCC codes is violated.}
    \label{Figure:privacy}
\end{figure}

A key advantage of GLCC and LCC codes is providing privacy guarantees on the dataset and models. Fig. \ref{Figure:privacy} shows the running time of GLCC codes with ($G=5,L=1$) and LCC codes with respect to the number of colluding workers $T$. While the running times of both LCC and GLCC codes increase with $T$, GLCC constantly achieves a faster running time for the same number of colluding workers. Moreover, for a fixed number of $50$ workers, GLCC codes tolerate up to $T=6$ colluding workers against privacy leakage, whereas the LCC codes can tolerate up to $T=3$ colluding workers. This demonstrates the superiority of GLCC codes in privacy protection over LCC codes.

As a benchmark, we exclusively tested the centralized training time at the master node, where the master locally performs the task of training the $5$ binary classifiers on the two image datasets, and there are neither encoding/decoding operations, nor communication, and nor straggler effect. It took 43.25 seconds for the MNIST dataset and 5829.28 seconds for the CIFAR-10 dataset. In comparison to Table \ref{tab:time}, we can observe that centralized training is significantly faster than distributed training for the MNIST dataset, while distributed training yields a notable improvement over centralized training for the CIFAR-10 dataset.
For the small MNIST dataset, the increase in distributed training time is attributed to the fact that the computational gains from distributed parallel computing cannot fully offset the time extension caused by communication, encoding, decoding, and stragglers. However, as the dataset size grows, as observed  in the case of CIFAR-10, the parallel gains become more substantial, thereby reducing the centralized training time.

\begin{figure}[htbp!]
\centering
	\subfigure[MNIST]{\begin{minipage}{0.49\linewidth}
		\centering
		\includegraphics[width=1.0\linewidth]{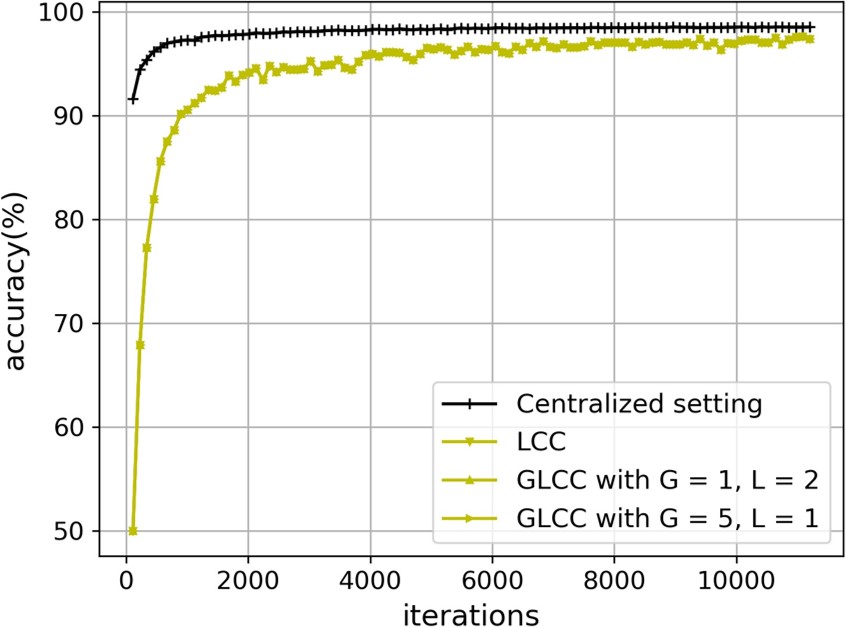}
		\label{Figure:accuracy:a}
	\end{minipage}}
	\subfigure[CIFAR-10]{\begin{minipage}{0.49\linewidth}
		\centering
		\includegraphics[width=1.0\linewidth]{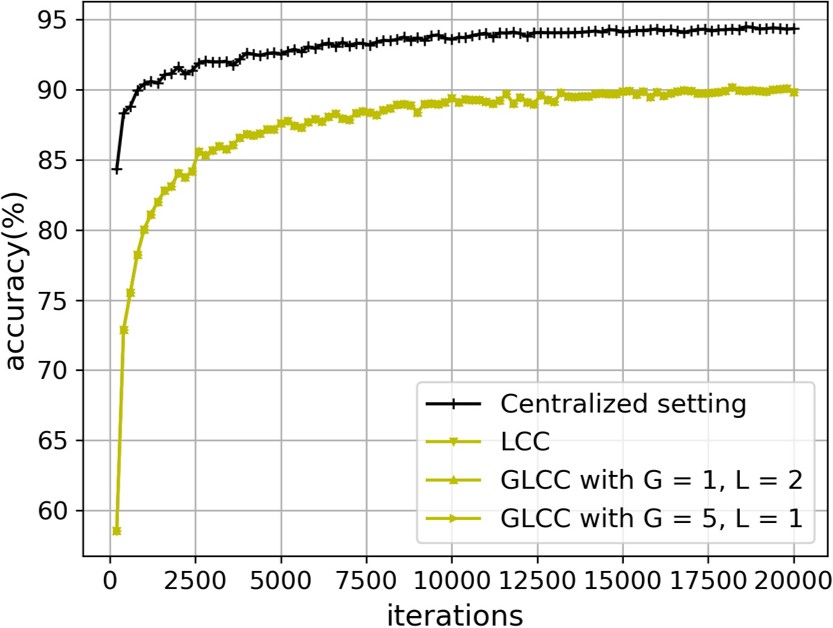}
		\label{Figure:accuracy:b}
	\end{minipage}}
    \caption{Comparison of the average test accuracy of GLCC codes, LCC codes, and centralized training across iterations.
}
    \label{Figure:accuracy}
\end{figure}

We also examine the accuracy of GLCC codes, LCC codes, and centralized training with increasing iteration rounds in Fig. \ref{Figure:accuracy}. We can obtain the following observations and results. 
\begin{itemize}
    \item GLCC codes and LCC codes achieve the same accuracy at the same number of iterations. This is because they both correctly complete the gradient computation for the sample batch. 
    \item Centralized training achieves an accuracy of $98.52\%$ and $94.36\%$ for MNIST and CIFAR-10, respectively. Compared to centralized training, GLCC/LCC codes experience a decrease of $1.14\%$ and $4.18\%$ in test accuracy for MNIST and CIFAR-10, respectively.    
    The main reason is the quantization losses incurred by the conversion from the real domain to the finite field. 
While increasing the values of the parameters $l_x$ and $l_w$ can reduce quantization losses, they also increase the risk of an overflow error incurred by wrap-around.
Similar to \cite{so2021codedprivateml,so2020scalable}, we perform the computational operations in a 64-bit implementation, such that the finite field size $q$ and the quantization parameters $l_x,l_w$ should not be chosen too large to avoid overflow errors, as indicated by Lemma \ref{size:field}. Consequently, GLCC/LCC codes result in accuracy losses, especially for large CIFAR-10 dataset.
\end{itemize}

Similar phenomena have also been observed in the case when applying GLCC/LCC codes to train more sophisticated machine learning models \cite{so2021codedprivateml,so2020scalable}, including Deep Neural Network (DNN) model, through polynomial approximations of nonlinear components such as the softmax and ReLU functions. To avoid overflow errors and reduce quantization losses, a larger size of the finite field is required. This requirement, in turn, will result in lower computation efficiency. Addressing this challenge poses an interesting avenue for future research. It is worth noting that if we do not consider providing data privacy guarantees, this problem is naturally resolved because there is no need for quantization operations.

\section{Conclusion}\label{conclusion}
We introduced GLCC codes for the problem of evaluating arbitrary multivariate polynomials in a distributed computing system, simultaneously achieving resiliency to straggling workers, security against adversarial workers, and privacy against colluding workers. Our proposed GLCC codes include the state-of-the-art LCC codes as a special case, and admit a more flexible tradeoff among recovery threshold, communication cost and computation complexity, which can be leveraged to optimize the computation latency of the system. 
We applied the GLCC codes into distributed training of single-layer perceptron neural networks, and experimentally demonstrated the substantial advantages of GLCC codes over baseline LCC codes, in speeding up the training process under multiple combinations of datasets and straggler scenarios.

\bibliographystyle{ieeetr}
\bibliography{reference.bib}

\begin{thebibliography}{10}

\bibitem{zhu2022generalized}
J.~Zhu and S.~Li, ``Generalized lagrange coded computing: A flexible computation-communication tradeoff,'' in {\em 2022 IEEE International Symposium on Information Theory (ISIT)}, pp.~832--837, IEEE, 2022.

\bibitem{Tail1}
J.~Dean and L.~A. Barroso, ``The tail at scale,'' {\em Communications of the ACM}, vol.~56, no.~2, pp.~74--80, 2013.

\bibitem{Tail3}
N.~J. Yadwadkar, B.~Hariharan, J.~E. Gonzalez, and R.~Katz, ``Multi-task learning for straggler avoiding predictive job scheduling,'' {\em The Journal of Machine Learning Research}, vol.~17, no.~1, pp.~3692--3728, 2016.

\bibitem{li2020coded}
S.~Li and S.~Avestimehr, ``Coded computing: Mitigating fundamental bottlenecks in large-scale distributed computing and machine learning,'' {\em Foundations and Trends{\textregistered} in Communications and Information Theory}, vol.~17, no.~1, 2020.

\bibitem{Lee}
K.~Lee, M.~Lam, R.~Pedarsani, D.~Papailiopoulos, and K.~Ramchandran, ``Speeding up distributed machine learning using codes,'' {\em IEEE Transactions on Information Theory}, vol.~64, no.~3, pp.~1514--1529, 2017.

\bibitem{li2016unified}
S.~Li, M.~A. Maddah-Ali, and A.~S. Avestimehr, ``A unified coding framework for distributed computing with straggling servers,'' in {\em 2016 IEEE Globecom Workshops (GC Wkshps)}, pp.~1--6, IEEE, 2016.

\bibitem{mallick2020rateless}
A.~Mallick, M.~Chaudhari, U.~Sheth, G.~Palanikumar, and G.~Joshi, ``Rateless codes for near-perfect load balancing in distributed matrix-vector multiplication,'' in {\em Abstracts of the 2020 SIGMETRICS/Performance Joint International Conference on Measurement and Modeling of Computer Systems}, pp.~95--96, 2020.

\bibitem{dutta2017coded}
S.~Dutta, V.~Cadambe, and P.~Grover, ``Coded convolution for parallel and distributed computing within a deadline,'' in {\em 2017 IEEE International Symposium on Information Theory (ISIT)}, pp.~2403--2407, IEEE, 2017.

\bibitem{Polynomialcode}
Q.~Yu, M.~A. Maddah-Ali, and A.~S. Avestimehr, ``Polynomial codes: an optimal design for high-dimensional coded matrix multiplication,'' in {\em Proceedings of the 31st International Conference on Neural Information Processing Systems}, pp.~4406--4416, 2017.

\bibitem{EPcode}
Q.~Yu, M.~A. Maddah-Ali, and A.~S. Avestimehr, ``Straggler mitigation in distributed matrix multiplication: Fundamental limits and optimal coding,'' {\em IEEE Transactions on Information Theory}, vol.~66, no.~3, pp.~1920--1933, 2020.

\bibitem{MatDotcode}
S.~Dutta, M.~Fahim, F.~Haddadpour, H.~Jeong, V.~Cadambe, and P.~Grover, ``On the optimal recovery threshold of coded matrix multiplication,'' {\em IEEE Transactions on Information Theory}, vol.~66, no.~1, pp.~278--301, 2019.

\bibitem{Tandonsecurecode}
W.-T. Chang and R.~Tandon, ``On the capacity of secure distributed matrix multiplication,'' in {\em 2018 IEEE Global Communications Conference (GLOBECOM)}, pp.~1--6, IEEE, 2018.

\bibitem{Kakar_secure_code}
J.~Kakar, S.~Ebadifar, and A.~Sezgin, ``On the capacity and straggler-robustness of distributed secure matrix multiplication,'' {\em IEEE Access}, vol.~7, pp.~45783--45799, 2019.

\bibitem{d2020gasp}
R.~G. D’Oliveira, S.~El~Rouayheb, and D.~Karpuk, ``Gasp codes for secure distributed matrix multiplication,'' {\em IEEE Transactions on Information Theory}, vol.~66, no.~7, pp.~4038--4050, 2020.

\bibitem{ZhuSDMM}
J.~Zhu, Q.~Yan, and X.~Tang, ``Improved constructions for secure multi-party batch matrix multiplication,'' {\em IEEE Transactions on Communications}, vol.~69, pp.~7673--7690, 2021.

\bibitem{zhu2020secure}
J.~Zhu and X.~Tang, ``Secure batch matrix multiplication from grouping lagrange encoding,'' {\em IEEE Communications Letters}, vol.~25, no.~4, pp.~1119--1123, 2020.

\bibitem{zhu2022systematic}
J.~Zhu and S.~Li, ``A systematic approach towards efficient private matrix multiplication,'' {\em IEEE Journal on Selected Areas in Information Theory}, vol.~3, no.~2, pp.~257--274, 2022.

\bibitem{jia2021cross}
Z.~Jia and S.~A. Jafar, ``Cross subspace alignment codes for coded distributed batch computation,'' {\em IEEE Transactions on Information Theory}, vol.~67, no.~5, pp.~2821--2846, 2021.

\bibitem{GradientCodes1}
R.~Tandon, Q.~Lei, A.~G. Dimakis, and N.~Karampatziakis, ``Gradient coding: Avoiding stragglers in distributed learning,'' in {\em International Conference on Machine Learning}, pp.~3368--3376, PMLR, 2017.

\bibitem{GradientCodes2}
N.~Raviv, I.~Tamo, R.~Tandon, and A.~G. Dimakis, ``Gradient coding from cyclic mds codes and expander graphs,'' {\em IEEE Transactions on Information Theory}, vol.~66, no.~12, pp.~7475--7489, 2020.

\bibitem{wang2019erasurehead}
H.~Wang, Z.~Charles, and D.~Papailiopoulos, ``Erasurehead: Distributed gradient descent without delays using approximate gradient coding,'' {\em arXiv preprint arXiv:1901.09671}, 2019.

\bibitem{li2018near}
S.~Li, S.~M.~M. Kalan, A.~S. Avestimehr, and M.~Soltanolkotabi, ``Near-optimal straggler mitigation for distributed gradient methods,'' in {\em 2018 IEEE International Parallel and Distributed Processing Symposium Workshops (IPDPSW)}, pp.~857--866, IEEE, 2018.

\bibitem{LiMapreduce}
S.~Li, M.~A. Maddah-Ali, Q.~Yu, and A.~S. Avestimehr, ``A fundamental tradeoff between computation and communication in distributed computing,'' {\em IEEE Transactions on Information Theory}, vol.~64, no.~1, pp.~109--128, 2017.

\bibitem{li2017scalable}
S.~Li, Q.~Yu, M.~A. Maddah-Ali, and A.~S. Avestimehr, ``A scalable framework for wireless distributed computing,'' {\em IEEE/ACM Transactions on Networking}, vol.~25, no.~5, pp.~2643--2654, 2017.

\bibitem{LCC}
Q.~Yu, S.~Li, N.~Raviv, S.~M.~M. Kalan, M.~Soltanolkotabi, and S.~A. Avestimehr, ``Lagrange coded computing: Optimal design for resiliency, security, and privacy,'' in {\em The 22nd International Conference on Artificial Intelligence and Statistics}, pp.~1215--1225, PMLR, 2019.

\bibitem{Qian_Yu}
Q.~Yu and A.~S. Avestimehr, ``Entangled polynomial codes for secure, private, and batch distributed matrix multiplication: Breaking the ``cubic'' barrier,'' in {\em 2020 IEEE International Symposium on Information Theory (ISIT)}, pp.~245--250, IEEE, 2020.

\bibitem{Strassen}
V.~Strassen, ``Gaussian elimination is not optimal,'' {\em Numerische mathematik}, vol.~13, no.~4, pp.~354--356, 1969.

\bibitem{Smirnov}
A.~V. Smirnov, ``The bilinear complexity and practical algorithms for matrix multiplication,'' {\em Computational Mathematics and Mathematical Physics}, vol.~53, no.~12, pp.~1781--1795, 2013.

\bibitem{soleymani2021list}
M.~Soleymani, R.~E. Ali, H.~Mahdavifar, and A.~S. Avestimehr, ``List-decodable coded computing: Breaking the adversarial toleration barrier,'' {\em IEEE Journal on Selected Areas in Information Theory}, vol.~2, no.~3, pp.~867--878, 2021.

\bibitem{tang2021verifiable}
T.~Tang, R.~E. Ali, H.~Hashemi, T.~Gangwani, S.~Avestimehr, and M.~Annavaram, ``Verifiable coded computing: Towards fast, secure and private distributed machine learning,'' {\em arXiv preprint arXiv:2107.12958}, 2021.

\bibitem{so2021codedprivateml}
J.~So, B.~G{\"u}ler, and A.~S. Avestimehr, ``Codedprivateml: A fast and privacy-preserving framework for distributed machine learning,'' {\em IEEE Journal on Selected Areas in Information Theory}, vol.~2, no.~1, pp.~441--451, 2021.

\bibitem{so2020scalable}
J.~So, B.~Guler, and S.~Avestimehr, ``A scalable approach for privacy-preserving collaborative machine learning,'' {\em Advances in Neural Information Processing Systems}, vol.~33, pp.~8054--8066, 2020.

\bibitem{so2022lightsecagg}
J.~So, C.~He, C.-S. Yang, S.~Li, Q.~Yu, R.~E~Ali, B.~Guler, and S.~Avestimehr, ``Lightsecagg: a lightweight and versatile design for secure aggregation in federated learning,'' {\em Proceedings of Machine Learning and Systems}, vol.~4, pp.~694--720, 2022.

\bibitem{shao2022dres}
J.~Shao, Y.~Sun, S.~Li, and J.~Zhang, ``Dres-fl: Dropout-resilient secure federated learning for non-iid clients via secret data sharing,'' {\em arXiv preprint arXiv:2210.02680}, 2022.

\bibitem{soleymani2021analog}
M.~Soleymani, H.~Mahdavifar, and A.~S. Avestimehr, ``Analog lagrange coded computing,'' {\em IEEE Journal on Selected Areas in Information Theory}, vol.~2, no.~1, pp.~283--295, 2021.

\bibitem{guruswami2013linear}
V.~Guruswami and C.~Wang, ``Linear-algebraic list decoding for variants of {R}eed--{S}olomon codes,'' {\em IEEE Transactions on Information Theory}, vol.~59, no.~6, pp.~3257--3268, 2013.

\bibitem{Shamir}
A.~Shamir, ``How to share a secret,'' {\em Communications of the ACM}, vol.~22, no.~11, pp.~612--613, 1979.

\bibitem{Lin}
S.~Lin and D.~Costello, {\em Error Control Coding: Fundamentals and Applications}.
\newblock Computer applications in electrical engineering series, Prentice-Hall, 1983.

\bibitem{Gao}
S.~Gao, ``A new algorithm for decoding {R}eed-{S}olomon codes,'' in {\em Communications, information and network security}, pp.~55--68, Springer, 2003.

\bibitem{ozfatura2021coded}
E.~Ozfatura, S.~Ulukus, and D.~G{\"u}nd{\"u}z, ``Coded distributed computing with partial recovery,'' {\em IEEE Transactions on Information Theory}, vol.~68, no.~3, pp.~1945--1959, 2021.

\bibitem{das2022coded}
A.~B. Das and A.~Ramamoorthy, ``Coded sparse matrix computation schemes that leverage partial stragglers,'' {\em IEEE Transactions on Information Theory}, vol.~68, no.~6, pp.~4156--4181, 2022.

\bibitem{kianidehkordi2020hierarchical}
S.~Kianidehkordi, N.~Ferdinand, and S.~C. Draper, ``Hierarchical coded matrix multiplication,'' {\em IEEE Transactions on Information Theory}, vol.~67, no.~2, pp.~726--754, 2020.

\bibitem{Von}
J.~Von Zur~Gathen and J.~Gerhard, {\em Modern computer algebra}.
\newblock Cambridge university press, 2013.

\bibitem{chen2008complexity}
N.~Chen and Z.~Yan, ``Complexity analysis of {R}eed-{S}olomon decoding over {GF}($2^m$) without using syndromes,'' {\em EURASIP Journal on Wireless Communications and Networking}, vol.~2008, pp.~1--11, 2008.

\bibitem{von2013modern}
J.~Von Zur~Gathen and J.~Gerhard, {\em Modern computer algebra}.
\newblock Cambridge university press, 2013.

\bibitem{van2022optimizing}
J.~van Der~Hoeven, ``Optimizing the half-gcd algorithm,'' {\em arXiv preprint arXiv:2212.12389}, 2022.

\bibitem{zhu2019deep}
L.~Zhu, Z.~Liu, and S.~Han, ``Deep leakage from gradients,'' {\em Advances in Neural Information Processing Systems}, vol.~32, 2019.

\bibitem{wang2019beyond}
Z.~Wang, M.~Song, Z.~Zhang, Y.~Song, Q.~Wang, and H.~Qi, ``Beyond inferring class representatives: User-level privacy leakage from federated learning,'' in {\em IEEE INFOCOM 2019-IEEE Conference on Computer Communications}, pp.~2512--2520, IEEE, 2019.

\bibitem{geiping2020inverting}
J.~Geiping, H.~Bauermeister, H.~Dr{\"o}ge, and M.~Moeller, ``Inverting gradients-how easy is it to break privacy in federated learning?,'' {\em Advances in Neural Information Processing Systems}, vol.~33, pp.~16937--16947, 2020.

\bibitem{so2020byzantine}
J.~So, B.~G{\"u}ler, and A.~S. Avestimehr, ``Byzantine-resilient secure federated learning,'' {\em IEEE Journal on Selected Areas in Communications}, vol.~39, no.~7, pp.~2168--2181, 2020.

\bibitem{lecun1998mnist}
Y.~LeCun, ``The {MNIST} database of handwritten digits,'' {\em http://yann. lecun. com/exdb/mnist/}, 1998.

\bibitem{krizhevsky2009learning}
A.~Krizhevsky, G.~Hinton, {\em et~al.}, ``Learning multiple layers of features from tiny images,'' 2009.

\bibitem{simonyan2014very}
K.~Simonyan and A.~Zisserman, ``Very deep convolutional networks for large-scale image recognition,'' {\em arXiv preprint arXiv:1409.1556}, 2014.

\bibitem{ozfatura2020straggler}
E.~Ozfatura, S.~Ulukus, and D.~G{\"u}nd{\"u}z, ``Straggler-aware distributed learning: Communication--computation latency trade-off,'' {\em Entropy}, vol.~22, no.~5, p.~544, 2020.

\end{thebibliography}

\end{document}